%% file: arxiv-np-2014-optlpdexp.tex
\documentclass[twocolumn,noshowpacs,noshowkeys,prb,aps,longbibliography,preprintnumbers]{revtex4-1}
\pdfoutput=0
\usepackage[caption=false]{subfig}
\usepackage{algorithm}
\usepackage{color}
\usepackage{algorithmic}
\usepackage{amsmath}
\usepackage{amssymb}
\usepackage{array}
\usepackage{url}
\usepackage{graphicx}
\usepackage{bibentry}

\providecommand{\U}[1]{\protect\rule{.1in}{.1in}}
\newtheorem{theorem}{Theorem}

\newtheorem{definition}[theorem]{Definition}

\newtheorem{lemma}[theorem]{Lemma}

\newenvironment{proof}[1][Proof]{\noindent\textbf{#1.} }{\ \rule{0.5em}{0.5em}}

\DeclareMathOperator{\trace}{Tr}

\newcommand{\ket}[1]{\left| #1 \right>} %
\newcommand{\bra}[1]{\left< #1 \right|} %

\newcommand\myatop[2]{\genfrac{}{}{0pt}{}{#1}{#2}} %

\begin{document}
\preprint{UMass Technical Report UM-CS-2014-004}
\title{Covert Optical Communication}
\author{Boulat A. Bash,$^{1,2}$ Andrei H. Gheorghe,$^{2,3}$ Monika Patel,$^2$\\Jonathan L. Habif,$^2$ Dennis Goeckel,$^4$ Don Towsley,$^1$ and Saikat Guha$^2$}
\affiliation{$^1$\textit{School of Computer Science, University of Massachusetts, Amherst, Massachusetts, USA 01003},\\
$^2$\textit{Quantum Information Processing Group, Raytheon BBN Technologies, Cambridge, Massachusetts, USA 02138},\\
$^3$\textit{Amherst College, Amherst, Massachusetts, USA 01002},\\
$^4$\textit{Electrical and Computer Engineering Department, University of Massachusetts, Amherst, Massachusetts, USA 01003}
}
\thanks{This material is based upon work supported by the National Science Foundation under Grants CNS-1018464 and ECCS-1309573. SG was supported by the aforesaid NSF grant, under subaward number 14-007829 A 00, and the DARPA {\em Information in a Photon} program under contract number HR0011-10-C-0159. BAB, AHG, JLH and MP would like to acknowledge financial support from Raytheon BBN Technologies.}

\begin{abstract}
\input{abstract}
\end{abstract}
\maketitle

\input{introduction}
\input{theory}
\input{experiments}

\input{conclusion}

\input{methods}

\bibliography{refs}

\widetext
\clearpage
\setcounter{equation}{0}
\setcounter{figure}{0}
\setcounter{table}{0}
\makeatletter
\renewcommand{\theequation}{S\arabic{equation}}
\renewcommand{\thefigure}{S\arabic{figure}}
\renewcommand{\bibnumfmt}[1]{[S#1]}
\section*{Supplementary Information}
\input{supplement}

\end{document}

%% file: abstract.tex
Encryption prevents unauthorized decoding, but does not ensure stealth---a security demand that a mere presence of a message be undetectable. We characterize the ultimate limit of covert communication that is secure against the most powerful physically-permissible adversary. We show that, although it is impossible over a pure-loss channel, covert communication is attainable in the presence of any excess noise, such as a $300$K thermal blackbody. In this case, $\mathcal{O}(\sqrt{n})$ bits can be transmitted reliably and covertly in $n$ optical modes using standard optical communication equipment. The all-powerful adversary may intercept all transmitted photons not received by the intended receiver, and employ arbitrary quantum memory and measurements. Conversely, we show that this square root scaling cannot be outperformed. We corroborate our theory in a proof-of-concept experiment. We believe that our findings will enable practical realizations of covert communication and sensing, both for point-to-point and networked scenarios.

%% file: introduction.tex
Encryption prevents unauthorized access to transmitted information---a security need 
critical to modern-day electronic communication. 
Conventional computationally-secure encryption~\cite{menezes96HAC,talb2006}, 
  information-theoretic secrecy~\cite{wyner, csiszar78secrecy}, 
  and quantum cryptography~\cite{bb84} offer progressively higher levels of 
  security. 
Quantum key distribution (QKD) allows two distant parties to generate shared
  secret keys over a lossy-noisy channel
  that are secure from the most powerful adversary allowed by physics. 
This shared secret, when subsequently used to encrypt data using the 
  one-time-pad cipher~\cite{shannon49sec}, yields the most powerful form of 
  encryption. 
However, encryption does not mitigate the threat to the users' 
  privacy from the discovery of the very existence of the
  message itself (e.g.,~seeking of ``meta-data'' as detailed in the recent
  Snowden disclosures \cite{bbc14snowden}), nor does it provide the means to
  communicate when the adversary forbids it.
Thus, low probability of detection (LPD) or \emph{covert} communication systems
  are desirable that not only protect the message content, but also prevent
  the detection of the transmission attempt. 
Here we delineate, and experimentally demonstrate, the ultimate limit
  of covert communication that is secure against the most powerful adversary 
  physically permissible---the same benchmark of security to which quantum 
  cryptography adheres for encrypted communication.

Covert communication is an ancient discipline~\cite{herodotus} revived by the
  communication revolution of the last century.
Modern developments include spread-spectrum radio-frequency (RF) 
  communication~\cite{simon94ssh}, where the signal power
  is suppressed below the noise floor by bandwidth expansion; and 
  steganography~\cite{fridrich09stego}, where messages are hidden in 
  fixed-size, finite-alphabet covertext objects such as digital images.
We recently characterized the information-theoretic limit of classical covert 
  communication on an additive white Gaussian noise (AWGN) channel, the 
  standard model for RF channels~\cite{bash13squarerootjsac,bash12sqrtlawisit}.
We showed that the sender Alice can reliably transmit $\mathcal{O}(\sqrt{n})$ 
  bits to the intended receiver Bob in $n$ AWGN channel uses with arbitrarily 
  low probability of detection by the adversary Willie.
Thus, 
  a non-trivial burst of 
  covert bits can be transmitted when $n$ is large. 
Our work was generalized to other channel settings 
  \cite{bash14timing,che13sqrtlawbscisit,kadhe14sqrtlawmultipathisit,hou14isit}.
Similar square-root laws were also found in 
  steganography\footnote{The $\log n$ improvement in steganographic application 
  versus covert communication over a noisy channel is 
  attributable to the noiseless Alice-to-Bob channel, and the similarity in the
  square root laws is due to the mathematics of 
  classical~\cite{lehmann05stathyp} and quantum~\cite{helstrom76quantumdetect} 
  statistical hypothesis testing.}, where it was shown that Alice can 
  modify $\mathcal{O}(\sqrt{n})$ symbols in a covertext of 
  size $n$, embedding $\mathcal{O}(\sqrt{n}\log n)$ hidden 
  bits~\cite{ker07pool,fridrich09stego,filler09sqrtlawmarkov,ker09sqrtlawkey,ker10sqrtlawnokey,shaw11qstego}. 

Optical signaling~\cite{bouchet10fso,senior09fiber} is particularly attractive
  for covert communication due to its narrow diffraction-limited beam spread 
  in free space~\cite{gagliardi95optcomms,goodman05fourieroptics} and the ease 
  of detecting fiber taps using time-domain reflectometry~\cite{anderson04otdr}.
Our information-theoretic analysis of covert communication on the AWGN channel 
  also applies to a lossy optical channel with additive Gaussian noise when 
  Alice uses a laser-light transmitter and both Bob and Willie use 
  coherent-detection receivers.
However, modern high-sensitivity optical communication components are primarily 
  limited by noise of quantum-mechanical origin. 
Thus, recent studies on the performance of physical optical communication 
  have focused on this quantum-limited regime~\cite{giovannetti04cappureloss,
  Wol98,Wil12}.
Here we establish the quantum limits of covert communication.
We demonstrate that covert communication is impossible over a pure-loss channel.
However, when the channel has any 
  excess noise (e.g.,~the unavoidable thermal noise from the blackbody 
  radiation at the operating temperature), Alice can reliably transmit 
  $\mathcal{O}(\sqrt{n})$ covert bits to Bob using $n$ optical modes,
  even if Willie intercepts all the photons not reaching Bob and employs 
  arbitrary quantum memory and measurements.
This is achievable using standard laser-light modulation and homodyne detection (thus the Alice-Bob channel is still an AWGN channel).
Thus, noise enables stealth.
Indeed, if Willie's detector contributes excess noise 
  (e.g.,~dark counts in photon-counting detectors), Alice can covertly
  communicate to Bob, even when the channel itself is pure-loss. 
We also show that the square-root limit cannot be outperformed.
We corroborate our theoretical results with a proof-of-concept 
  experiment, where the excess noise in Willie's detection is emulated by dark 
  counts of his single photon detector. 
This is the first known implementation of a truly 
  quantum-information-theoretically secure covert communication system that 
  allows communication when all transmissions are prohibited.

\section*{Information-Theoretically Covert Communication}
Quantum and classical information-theoretic analyses of covert communication
  consider the \emph{reliability} and \emph{detectability} of a transmission.
We introduce these concepts next.
\\

\noindent {\em Reliability}---We consider a scenario where Alice attempts to 
  transmit $M$ bits to Bob using $n$ optical modes while Willie 
  attempts to detect her transmission attempt.
Each of the $2^M$ possible $M$-bit messages maps to an $n$-mode
  \emph{codeword}, and their collection forms a \emph{codebook}.
Since we consider single-spatial-mode fiber and free-space optical channels,
  each of the $n$ modes in the codeword corresponds to a signaling interval
  carrying one modulation symbol.
Desirable codebooks ensure that the codewords, when corrupted by the 
  channel, are distinguishable from one another.
This provides \emph{reliability}: a guarantee that the probability
  of Bob's error in decoding Alice's message $\mathbb{P}_e^{(b)}<\delta$
  with arbitrarily small $\delta>0$ for $n$ large enough.
In practice, error-correction codes (ECCs) are used to enable reliability.
\\

\noindent {\em Detectability}---
Willie's detector reduces to a binary hypothesis test of Alice's transmission
  state given his observations of the channel.
Denote by $\mathbb{P}_{\mathrm{FA}}$ the probability that Willie raises a false alarm 
  when Alice does not transmit, and by $\mathbb{P}_{\mathrm{MD}}$ the probability that 
  Willie misses the detection of Alice's transmission.
Under the assumption of equal prior
  probabilities on Alice's transmission state (unequal prior
  probabilities do not affect the asymptotics
  \cite{bash13squarerootjsac}), Willie's \emph{detection error
  probability}, 
  $\mathbb{P}_e^{(w)}=(\mathbb{P}_{\mathrm{FA}}+\mathbb{P}_{\mathrm{MD}})/2$.
Alice desires a reliable signaling scheme that is \emph{covert}, i.e., ensures
  $\mathbb{P}_e^{(w)}\geq1/2-\epsilon$ for an arbitrarily small 
  $\epsilon>0$ regardless of Willie's quantum measurement choice
  (since $\mathbb{P}_e^{(w)}=1/2$ for a random guess).
By decreasing her transmission power, Alice can decrease the effectiveness of
  Willie's hypothesis test at the expense of the reliability of Bob's decoding. 
\emph{Information-theoretically secure covert communication} is both
  reliable and covert. 
To achieve it, prior to transmission, Alice and Bob share a
  secret, the cost of which we assume to be substantially less than that of 
  being detected by Willie. 
Secret-sharing is consistent with other information-hiding systems
  \cite{bash13squarerootjsac,bash12sqrtlawisit,ker07pool,fridrich09stego,filler09sqrtlawmarkov,ker09sqrtlawkey,ker10sqrtlawnokey,shaw11qstego};
  however, as evidenced by the recent results for a restricted class of
  channels~\cite{che13sqrtlawbscisit,kadhe14sqrtlawmultipathisit}, we believe
  that certain scenarios (e.g.,~Willie's channel from Alice being worse than 
  Bob's) will allow secret-less optical covert communication.

%% file: theory.tex
\section*{Analysis of Covert Optical Communication}

Here we outline the theoretical development of quantum-information-theoretically
  secure covert optical communication. 
Formal theorem statements are deferred to the Methods, with detailed proofs
  in the Supplementary Information.
\\

\begin{figure}
\centering
\includegraphics[width=\columnwidth]{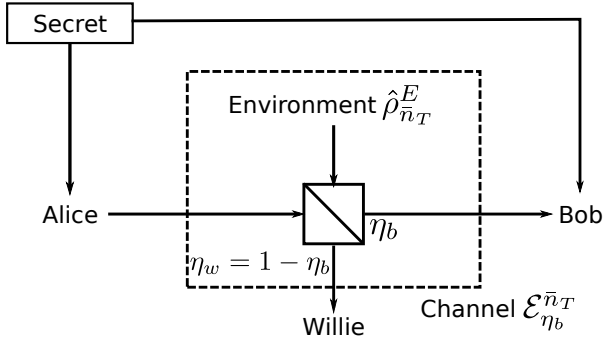}
\caption{Channel model. 
The input-output relationship is captured by a beamsplitter of transmissivity 
  $\eta_b$, with the transmitter Alice at one of the input ports and the 
  intended receiver Bob at one of the output ports, and
  $\eta_b$ being the fraction of Alice's signaling photons that reach Bob.
The other input and output ports of the beamsplitter correspond to the 
  environment and the adversary Willie.
Willie collects the entire $\eta_w=1-\eta_b$ fraction of Alice's 
  photons that do not reach Bob.
This models single-spatial-mode free-space and single-mode fiber optical 
  channels. 
Alice and Bob share a secret before the transmission.
\label{fig:channel}}
\end{figure}

\noindent {\em Channel model}---Consider a single-mode quasi-monochromatic 
  lossy optical channel ${\cal E}_{\eta_b}^{\bar{n}_T}$ of transmissivity
  $\eta_b \in (0,1]$ and thermal noise mean photon number per mode
  $\bar{n}_T \ge 0$, as depicted in Figure 
  \ref{fig:channel}.
Willie collects the entire $\eta_w=1-\eta_b$ fraction of Alice's 
  photons that do not reach Bob but otherwise remains passive, not injecting
  any light into the channel. 
Later we argue that being active does not help Willie to 
  detect Alice's transmissions.
For a pure loss channel ($\bar{n}_T=0$), the environment input is in the 
  \emph{vacuum} state $\hat{\rho}^E_0=|0\rangle\langle0|^E$, corresponding 
  to the minimum noise the channel must inject to preserve the 
  Heisenberg inequality of quantum mechanics. 
\\

\noindent{\em Pure loss insufficient for covert communication}---Regardless of 
  Alice's strategy, reliable and covert communication over a pure-loss channel
  to Bob is impossible. 
Theorem \ref{th:pureloss} in the Methods demonstrates that Willie can 
  effectively use an ideal \emph{single photon detector} (SPD) on each mode
  to discriminate between an $n$-mode vacuum state and any non-vacuum state in
  Alice's codebook.
Willie avoids false alarms since no photons impinge on his SPD when Alice is 
  silent.
However, a single \emph{click}---detection of one or more photons---gives away 
  Alice's transmission attempt regardless of the actual quantum state of Alice's
  signaling photons.
Alice is thus constrained to codewords that are nearly indistinguishable from 
  vacuum, rendering unreliable any communication attempt that is designed to be
  covert.
Furthermore, any communication attempt that is designed to be reliable 
  cannot remain covert, as Willie detects it with high 
  probability for large $n$. 
This is true even when Alice and Bob have access to an infinitely-large
  pre-shared secret.
Thus, if Willie controlled the environment (as assumed in QKD proofs), by
  setting it to vacuum, he could deny covert communication between Alice and 
  Bob.
However, a positive amount of non-adversarial excess noise---whether from 
  the thermal background or the detector itself---is unavoidable, which
  enables covert communication.
\\

\noindent{\em Channel noise yields the square root law}---Now consider 
  the lossy bosonic channel 
  ${\cal E}_{\eta_b}^{\bar{n}_T}$, where the environment mode is in a thermal 
  state with mean photon number $\bar{n}_T>0$. 
A thermal state is represented by a mixture of coherent states 
  $\ket{\alpha}$---quantum descriptors of ideal laser-light---weighted by a 
  Gaussian distribution over the field amplitude $\alpha \in {\mathbb C}$, 
  $\hat{\rho}^E_{\bar{n}_T}=\frac{1}{\pi \bar{n}_T}\int e^{-{|\alpha|^2}/{\bar{n}_T}}|\alpha\rangle\langle\alpha |^E{\rm d}^2\alpha$. 
This thermal noise masks Alice's transmission attempt, enabling covert 
  communication even when Willie has arbitrary resources, such as
  access to all signaling photons not captured by Bob and
  any quantum-limited measurement on the light he thus captures. 
Theorem~\ref{th:thermal} in the Methods demonstrates that in this scenario
  Alice can reliably transmit $\mathcal{O}(\sqrt{n})$ covert bits using $n$ 
  optical modes to Bob, who needs only a conventional homodyne-detection 
  receiver.
Alice achieves this using mean photon number per mode
  $\bar{n}=\mathcal{O}(1/\sqrt{n})$.
Conversely, Theorem~\ref{th:converse} states that if Alice
  exceeds the limit of $\mathcal{O}(\sqrt{n})$ covert bits in $n$ optical modes,
  transmission either is detected or unreliable.
\\
 
\noindent{\em Detector noise also enables covert communication}---While any 
  $\bar{n}_T>0$ enables covert communication, the number of 
  covertly-transmitted bits decreases with $\bar{n}_T$. 
Blackbody radiation is negligible at optical frequencies (e.g., a typical 
  daytime value of $\bar{n}_T \approx 10^{-6}$ photons per mode at the optical
  telecom wavelength of $1.55 \mu$m~\cite{Kop70}).
However, other sources of excess noise can also hide the transmissions
  (e.g.~detector dark counts and Johnson noise).
To illustrate information-hiding capabilities of these noise sources, we 
  consider the (hypothetical) pure-loss channel, for which
  Willie's optimal receiver is an 
  ideal photon number resolving (PNR) detector on each mode
  (as discussed in the Supplementary Information).
The prevalent form of excess noise afflicting PNR detectors is the \emph{dark 
  counts}---erroneous detection events stemming from internal 
  spontaneous emission processes.
Thus, we consider a pure-loss channel where
  Willie is equipped with a PNR detector. %
Theorem \ref{th:dark} in the Methods demonstrates that, using an on-off keying 
  (OOK) coherent state modulation where Alice transmits the {\em on} symbol 
  $\ket{\alpha}$ with probability $q=\mathcal{O}(1/\sqrt{n})$ and the {\em off} 
  symbol $\ket{0}$ with probability $1-q$, Alice can reliably transmit 
  $\mathcal{O}(\sqrt{n})$ covert bits using $n$ OOK symbols. 
\\

\noindent{\em A structured strategy for covert communication}---The skewed 
  on-off duty cycle of OOK modulation makes construction of
  efficient ECCs challenging. 
Constraining OOK signaling to $Q$-ary pulse position modulation (PPM)
  addresses this issue by sacrificing a constant fraction of throughput.
Each PPM symbol uses a PPM {\em frame} to transmit a sequence of $Q$ coherent 
  state pulses, $\ket{0}\ldots\ket{\alpha}\ldots\ket{0}$, encoding message 
  $i\in\{1,2,\ldots,Q\}$ by transmitting $\ket{\alpha}$ in the 
  $i^{\text{th}}$ mode of the PPM frame. 
Thus, instead of $\mathcal{O}(n)$ bits allowed by OOK, PPM lets 
  $\mathcal{O}\left(\frac{n\log Q}{Q}\right)$ bits be transmitted in $n$
  optical modes.
However, PPM performs well in the low photon number regime~\cite{wang12ppm} 
  and the symmetry of its symbols enables the use of many efficient ECCs.

To communicate covertly, Alice and Bob use a fraction 
  $\zeta=\mathcal{O}\left(\sqrt{Q/n}\right)$ of $n/Q$ available 
  PPM frames on average, effectively using $\bar{n}=\mathcal{O}(1/\sqrt{n})$
  photons per mode.
By keeping secret which frames they use, Alice and Bob force Willie to 
  examine all of them, increasing the likelihood of dark counts.
An ECC that is known by Willie ensures reliability.
However, the transmitted pulse positions
  are scrambled within the corresponding PPM frames via
  an operation resembling one-time pad encryption~\cite{shannon49sec},
  preventing Willie's exploitation of the ECC's structure for detection
  (rather than protecting the message content).
Theorem \ref{th:lpd_ppm} demonstrates that, using this scheme, Alice reliably
  transmits
  $\mathcal{O}\left(\sqrt{\frac{n}{Q}}\log Q\right)$ covert bits at the cost 
  of pre-sharing $\mathcal{O}\left(\sqrt{\frac{n}{Q}}\log n\right)$ secret bits.

%% file: experiments.tex
\section*{Experimental Results}
\noindent\emph{Objective and design}---To demonstrate the square-root law of 
  covert optical communication we realized a proof-of-concept test-bed 
  implementation.
Alice and Bob engage in an $n$-mode communication 
  session consisting of $n/Q$ $Q$-ary PPM frames, $Q=32$.
As described in the Methods, Alice transmits $\zeta n/Q$ PPM symbols on average,
  using a first order Reed-Solomon (RS) code for error correction. 
RS codes perform well on channels dominated by {\em erasures},
  which occur in low received-power scenarios, e.g.,~covert and deep space
  communication~\cite{moision06deepspace}.
We defer the specifics of the generation of the transmitted signal to
  the Methods.
We varied $n$ from $3.2\times10^6$ to $3.2\times10^7$ in several communication 
  regimes: ``careful Alice'' ($\zeta=0.25\sqrt{Q/n}$),
  ``careless Alice'' ($\zeta=0.03\sqrt[4]{Q/n}$), and ``dangerously
  careless Alice'' ($\zeta=0.003$ and $\zeta=0.008$).  
For each $(n,\zeta)$ pair we conducted $100$ experiments and $10^5$ Monte-Carlo
  simulations, measuring Bob's total number of bits 
  received and Willie's detection error probability.
\\

\begin{figure}
\centering
\includegraphics[width=\columnwidth]{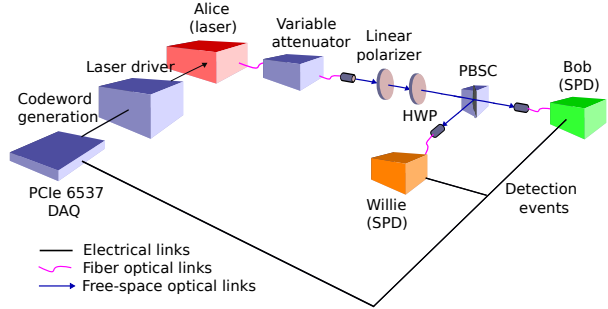}
\caption{Experimental setup. 
A National Instruments PCIe-6537 data acquisition card (DAQ), driven by a 1 MHz 
  clock, controlled the experiment, generating 
  transmissions and reading detection events.
Alice generated 1 ns optical pulses using
  a temperature-stabilized laser diode with center wavelength 1550.2 nm.
The pulses were sent into a free-space optical channel, where a half-wave plate 
  (HWP) and polarizing beamsplitter cube (PBSC) sent a fraction $\eta_b$ 
  of light to Bob, and the remaining light to Willie.  
Bob and Willie's receivers operated InGaAs 
  Geiger-mode avalanche photodiode SPDs that 
  were gated with 1 ns reverse bias triggered to match the
  arrival of Alice's pulses.
\label{fig:Exp_setup}}
\end{figure}

\input{channel_chars_tab}

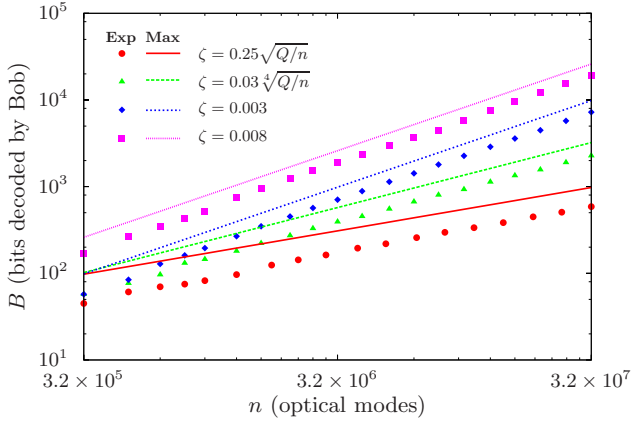
\begin{figure}[t]
\centering
\begin{center}
\resizebox{\columnwidth}{!}{\input{bob.tex}}
\end{center}
\vspace{-0.2in}
\caption{Bits decoded by Bob.  
  Each data point is an average from $100$ experiments, with negligibly small 
  95\% confidence intervals. 
  The symbol error rates are: $1.1\times10^{-4}$ for $\zeta=0.25\sqrt{Q/n}$, 
  $8.3\times10^{-3}$ for $\zeta=0.03/\sqrt[4]{Q/n}$, $4.5\times10^{-3}$ for 
  $\zeta=0.003$, and $1.8\times10^{-2}$ for $\zeta=0.008$. 
  We also report the maximum throughput $\frac{C_s\zeta n}{Q}$ computed in the 
  Methods using the experimentally-observed values from 
  Table~\ref{tab:exp_params}, where $C_s$ is
  the per-symbol Shannon capacity~\cite{shannon48it}.  
  Given the low observed symbol error rate for $\zeta=0.25\sqrt{Q/n}$, we note 
  that a square root scaling is achievable even using a relatively short RS 
  code; Figure~\ref{fig:willie} demonstrates that this is achieved covertly.
\label{fig:bits}}
\end{figure}

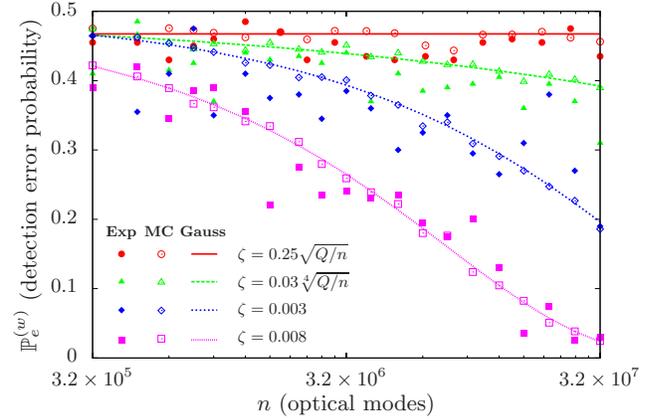
\begin{figure}[t]
\centering
\begin{center}
\resizebox{\columnwidth}{!}{\input{willie.tex}}
\end{center}
\vspace{-0.2in}
\caption{Willie's error probability.  
  Estimates from $100$ experiments have solid fill; estimates from $10^5$ 
  Monte-Carlo simulations have clear fill; and Gaussian approximations are 
  lines. 
  The 95\% confidence intervals (computed in the Methods) for the experimental
  estimates are $\pm 0.136$; for the Monte-Carlo simulations they are 
  $\pm 0.014$. 
  Alice transmits $\zeta n/Q$ PPM symbols on average and
  Willie's error probability remains constant when Alice obeys the square root 
  law and uses $\zeta=\mathcal{O}(\sqrt{Q/n})$; it drops as $n$ increases if 
  Alice breaks the square root law by using an asymptotically larger $\zeta$.
\label{fig:willie}}
\end{figure}

\noindent\emph{Implementation}---The experiment was conducted using a mixture of
  fiber-based and free-space
  optical elements implementing channels from Alice to both Bob and Willie  
  (see Figure \ref{fig:Exp_setup} for a schematic).
Due to the low intensity of Alice's pulses, direct detection using single
  photon detectors (SPDs), rather than PNR receivers, was sufficient.
Several configurations were considered for implementing the background noise
  at the receivers.  
We provided noise only during the gating period of the detectors since
  continuous wave light irradiating Geiger-mode avalanche photodiodes (APDs) 
  suppresses detection efficiency \cite{Ma09}.
Instead of providing extraneous optical pulses during the gating window of the 
  APD, we emulated optical noise at the detectors by 
  increasing the detector gate voltage, thus increasing the detector's dark 
  click probability.  
While the APD dark counts are Poisson-distributed with mean rate $\bar{n}_N$
  photons per mode, when $\bar{n}_N \ll 1$, the dark click probability
  $1-e^{-\bar{n}_N}$ is close to $\frac{\bar{n}_N}{1+\bar{n}_N}$, 
  the probability that an incoherent thermal 
  background with mean photon number per mode $\bar{n}_N$ produces a click.
In Table \ref{tab:exp_params} we report the experimentally-observed and 
  targeted values of dark click probabilities $p_D^{(b)}$ and $p_D^{(w)}$ of 
  Bob's and Willie's detectors, as well as the mean number of photons detected 
  by Bob $\bar{n}^{(b)}_{\mathit{det}}=\eta_{b}\eta_{\mathrm{QE}}^{(b)}\bar{n}$
  and Willie 
  $\bar{n}^{(w)}_{\mathit{det}} = (1-\eta_{b})\eta_{\mathrm{QE}}^{(w)}\bar{n}$, 
  where $\bar{n}=5$ is the mean photon number of Alice's pulses, $\eta_b=0.97$
  is the fraction of light sent to Bob, and $\eta_{\mathrm{QE}}^{(b)}$ and 
  $\eta_{\mathrm{QE}}^{(w)}$ 
  are the quantum efficiencies of Bob's and Willie's detectors, which
  we do not explicitly calculate.
However, quantum efficiency is strongly correlated with the detector's dark 
  click probability~\cite{Ri98}.

The amount of transmitted information, with other parameters fixed, is 
  proportional to $\bar{n}_{\mathit{det}}^{(b)}/\bar{n}_{\mathit{det}}^{(w)}$. 
Our choice of $\bar{n}_{\mathit{det}}^{(b)} \gg \bar{n}_{\mathit{det}}^{(w)}$ 
  allowed the experiment to gather a statistically meaningful data sample in 
  a reasonable duration.  
In an operational free-space laser communication system, a directional 
  transmitter will likely yield just such an asymmetry in coupling between Bob 
  and Willie; however, we note that the only fundamental requirement
  for implementing information-theoretically secure covert communication is 
  $p_D^{(w)}>0$, or $\bar{n}_T>0$.
\\

\noindent\emph{Results}---Alice and Bob
  use a $(31,15)$ RS code.
Figure \ref{fig:bits} reports the number of bits received by Bob with the 
  corresponding symbol error rate in our 
  experiments, and his maximum throughput from Alice
  (calculated for each regime using the experimentally-observed values from 
  Table \ref{tab:exp_params}).
The details of our analysis are in the Methods.
Our relatively short RS code achieves between 45\% and 60\% of the maximum  
  throughput in the ``careful Alice'' regime and between 55\% and 75\% of  
  the maximum in other regimes at reasonable error rates,
  showing that even a basic code demonstrates our theoretical scaling.

Willie's detection problem can be reduced to a test between two simple
  hypotheses where the log-likelihood ratio test minimizes 
  $\mathbb{P}_e^{(w)}$~\cite{lehmann05stathyp}.
Figure \ref{fig:willie} reports Willie's probability of error estimated 
  from the experiments and the Monte-Carlo study,
  as well as its analytical Gaussian approximation, with the implementation
  details deferred to the Methods.
Monte-Carlo simulations show that the Gaussian approximation is accurate.
More importantly, Figure \ref{fig:willie} highlights Alice's safety when she
  obeys the square root law and her peril when she does not.
When $\zeta=\mathcal{O}(1/\sqrt{n})$, $\mathbb{P}_e^{(w)}$ remains 
  constant as $n$ increases.
However, for asymptotically larger $\zeta$, $\mathbb{P}_e^{(w)}$ drops at a rate
  that depends on Alice's carelessness. 
The drop at $\zeta=0.008$ vividly demonstrates our converse.

%% file: channel_chars_tab.tex
\begin{table}
\caption{Optical channel characteristics}
\label{tab:exp_params}
\ifx\naturesub\undefined
\begin{tabular}{|m{2.3cm}||c|c||c|c|}
\else
\begin{tabular}{|l||c|c||c|c|}
\fi
\hline
&\multicolumn{2}{c||}{Willie}&\multicolumn{2}{c|}{Bob}\\
\hline
Experimental observations& $p_D^{(w)}$ & $\bar{n}_{\mathit{det}}^{(w)}$ & $p_D^{(b)}$ & $\bar{n}_{\mathit{det}}^{(b)}$ \\
\hline
$\zeta=0.25\sqrt{Q/n}$ & $9.15\times10^{-5}$ & $0.036$ & $2.99\times10^{-6}$ & $1.52$ \\
$\zeta=0.03\sqrt[4]{Q/n}$ & $9.11\times10^{-5}$ & $0.032$ & $2.55\times10^{-6}$ & $1.14$ \\
$\zeta=0.003$ & $9.29\times10^{-5}$ & $0.032$ & $2.65\times10^{-6}$ & $1.07$ \\
$\zeta=0.008$ & $9.27\times10^{-5}$ & $0.028$ & $2.68\times10^{-6}$ & $1.05$ \\
\hline
\hline
Target: & $9\times10^{-5}$ & $0.03$ & $3\times10^{-6}$ & $1.4$ \\
\hline
\end{tabular}
\end{table}

%% file: bob.tex
\begingroup
  \selectfont
  \makeatletter
  \providecommand\color[2][]{%
    \GenericError{(gnuplot) \space\space\space\@spaces}{%
      Package color not loaded in conjunction with
      terminal option `colourtext'%
    }{See the gnuplot documentation for explanation.%
    }{Either use 'blacktext' in gnuplot or load the package
      color.sty in LaTeX.}%
    \renewcommand\color[2][]{}%
  }%
  \providecommand\includegraphics[2][]{%
    \GenericError{(gnuplot) \space\space\space\@spaces}{%
      Package graphicx or graphics not loaded%
    }{See the gnuplot documentation for explanation.%
    }{The gnuplot epslatex terminal needs graphicx.sty or graphics.sty.}%
    \renewcommand\includegraphics[2][]{}%
  }%
  \providecommand\rotatebox[2]{#2}%
  \@ifundefined{ifGPcolor}{%
    \newif\ifGPcolor
    \GPcolortrue
  }{}%
  \@ifundefined{ifGPblacktext}{%
    \newif\ifGPblacktext
    \GPblacktexttrue
  }{}%
  \let\gplgaddtomacro\g@addto@macro
  \gdef\gplbacktext{}%
  \gdef\gplfronttext{}%
  \makeatother
  \ifGPblacktext
    \def\colorrgb#1{}%
    \def\colorgray#1{}%
  \else
    \ifGPcolor
      \def\colorrgb#1{\color[rgb]{#1}}%
      \def\colorgray#1{\color[gray]{#1}}%
      \expandafter\def\csname LTw\endcsname{\color{white}}%
      \expandafter\def\csname LTb\endcsname{\color{black}}%
      \expandafter\def\csname LTa\endcsname{\color{black}}%
      \expandafter\def\csname LT0\endcsname{\color[rgb]{1,0,0}}%
      \expandafter\def\csname LT1\endcsname{\color[rgb]{0,1,0}}%
      \expandafter\def\csname LT2\endcsname{\color[rgb]{0,0,1}}%
      \expandafter\def\csname LT3\endcsname{\color[rgb]{1,0,1}}%
      \expandafter\def\csname LT4\endcsname{\color[rgb]{0,1,1}}%
      \expandafter\def\csname LT5\endcsname{\color[rgb]{1,1,0}}%
      \expandafter\def\csname LT6\endcsname{\color[rgb]{0,0,0}}%
      \expandafter\def\csname LT7\endcsname{\color[rgb]{1,0.3,0}}%
      \expandafter\def\csname LT8\endcsname{\color[rgb]{0.5,0.5,0.5}}%
    \else
      \def\colorrgb#1{\color{black}}%
      \def\colorgray#1{\color[gray]{#1}}%
      \expandafter\def\csname LTw\endcsname{\color{white}}%
      \expandafter\def\csname LTb\endcsname{\color{black}}%
      \expandafter\def\csname LTa\endcsname{\color{black}}%
      \expandafter\def\csname LT0\endcsname{\color{black}}%
      \expandafter\def\csname LT1\endcsname{\color{black}}%
      \expandafter\def\csname LT2\endcsname{\color{black}}%
      \expandafter\def\csname LT3\endcsname{\color{black}}%
      \expandafter\def\csname LT4\endcsname{\color{black}}%
      \expandafter\def\csname LT5\endcsname{\color{black}}%
      \expandafter\def\csname LT6\endcsname{\color{black}}%
      \expandafter\def\csname LT7\endcsname{\color{black}}%
      \expandafter\def\csname LT8\endcsname{\color{black}}%
    \fi
  \fi
  \setlength{\unitlength}{0.0500bp}%
  \begin{picture}(7200.00,5040.00)%
    \gplgaddtomacro\gplbacktext{%
      \csname LTb\endcsname%
      \put(536,755){\makebox(0,0)[r]{\strut{}\scalebox{1.35}{$10^1$}}}%
      \put(536,1801){\makebox(0,0)[r]{\strut{}\scalebox{1.35}{$10^2$}}}%
      \put(536,2847){\makebox(0,0)[r]{\strut{}\scalebox{1.35}{$10^3$}}}%
      \put(536,3892){\makebox(0,0)[r]{\strut{}\scalebox{1.35}{$10^4$}}}%
      \put(536,4938){\makebox(0,0)[r]{\strut{}\scalebox{1.35}{$10^5$}}}%
      \put(720,500){\makebox(0,0){\strut{}\scalebox{1.35}{$3.2\times 10^5$}}}%
      \put(3780,500){\makebox(0,0){\strut{}\scalebox{1.35}{$3.2\times 10^6$}}}%
      \put(6839,500){\makebox(0,0){\strut{}\scalebox{1.35}{$3.2\times 10^7$}}}%
      \put(-62,2846){\rotatebox{-270}{\makebox(0,0){\strut{}\scalebox{1.5}{$B$ (bits decoded by Bob)}}}}%
      \put(3779,194){\makebox(0,0){\strut{}\scalebox{1.5}{$n$ (optical modes)}}}%
      \put(984,4632){\makebox(0,0)[l]{\strut{}\textbf{Exp}}}%
      \put(1464,4632){\makebox(0,0)[l]{\strut{}\textbf{Max}}}%
    }%
    \gplgaddtomacro\gplfronttext{%
      \csname LTb\endcsname%
      \put(1832,4453){\makebox(0,0)[l]{\strut{}\scalebox{1.1}{~~~~$\zeta=0.25\sqrt{Q/n}$}}}%
      \csname LTb\endcsname%
      \put(1832,4113){\makebox(0,0)[l]{\strut{}\scalebox{1.1}{~~~~$\zeta=0.03\sqrt[4]{Q/n}$}}}%
      \csname LTb\endcsname%
      \put(1832,3773){\makebox(0,0)[l]{\strut{}\scalebox{1.1}{~~~~$\zeta=0.003$}}}%
      \csname LTb\endcsname%
      \put(1832,3433){\makebox(0,0)[l]{\strut{}\scalebox{1.1}{~~~~$\zeta=0.008$}}}%
    }%
    \gplbacktext
    \put(0,0){\includegraphics{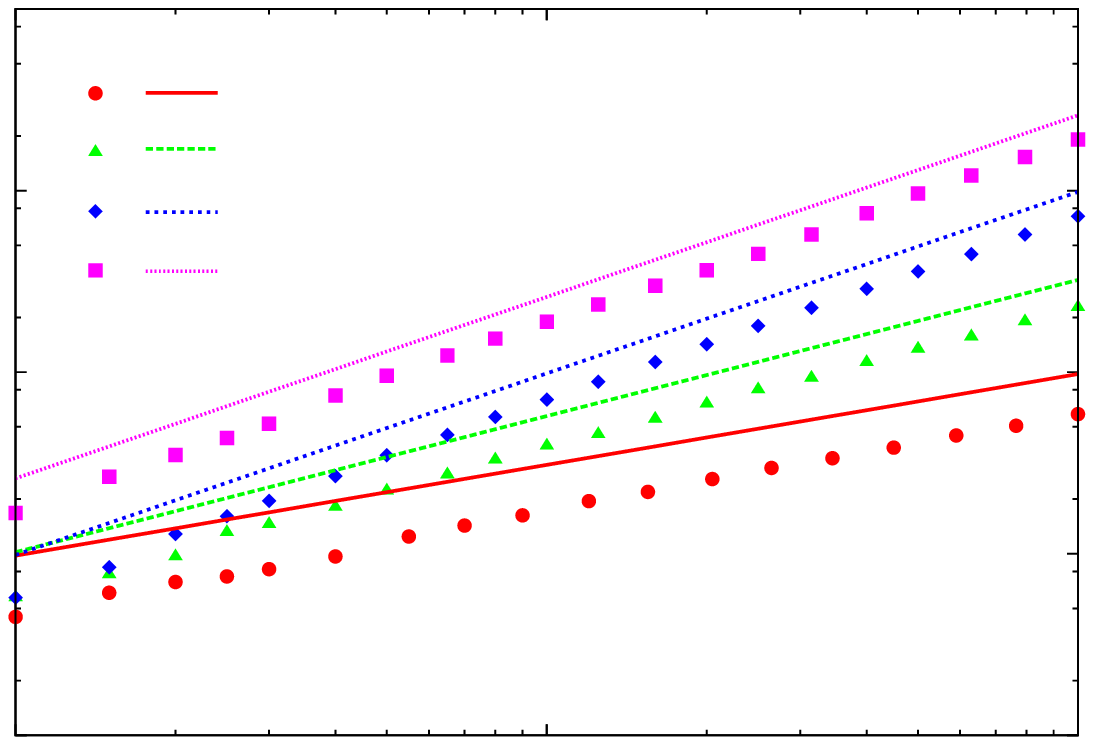}}%
    \gplfronttext
  \end{picture}%
\endgroup

%% file: willie.tex
\begingroup
  \selectfont
  \makeatletter
  \providecommand\color[2][]{%
    \GenericError{(gnuplot) \space\space\space\@spaces}{%
      Package color not loaded in conjunction with
      terminal option `colourtext'%
    }{See the gnuplot documentation for explanation.%
    }{Either use 'blacktext' in gnuplot or load the package
      color.sty in LaTeX.}%
    \renewcommand\color[2][]{}%
  }%
  \providecommand\includegraphics[2][]{%
    \GenericError{(gnuplot) \space\space\space\@spaces}{%
      Package graphicx or graphics not loaded%
    }{See the gnuplot documentation for explanation.%
    }{The gnuplot epslatex terminal needs graphicx.sty or graphics.sty.}%
    \renewcommand\includegraphics[2][]{}%
  }%
  \providecommand\rotatebox[2]{#2}%
  \@ifundefined{ifGPcolor}{%
    \newif\ifGPcolor
    \GPcolortrue
  }{}%
  \@ifundefined{ifGPblacktext}{%
    \newif\ifGPblacktext
    \GPblacktexttrue
  }{}%
  \let\gplgaddtomacro\g@addto@macro
  \gdef\gplbacktext{}%
  \gdef\gplfronttext{}%
  \makeatother
  \ifGPblacktext
    \def\colorrgb#1{}%
    \def\colorgray#1{}%
  \else
    \ifGPcolor
      \def\colorrgb#1{\color[rgb]{#1}}%
      \def\colorgray#1{\color[gray]{#1}}%
      \expandafter\def\csname LTw\endcsname{\color{white}}%
      \expandafter\def\csname LTb\endcsname{\color{black}}%
      \expandafter\def\csname LTa\endcsname{\color{black}}%
      \expandafter\def\csname LT0\endcsname{\color[rgb]{1,0,0}}%
      \expandafter\def\csname LT1\endcsname{\color[rgb]{0,1,0}}%
      \expandafter\def\csname LT2\endcsname{\color[rgb]{0,0,1}}%
      \expandafter\def\csname LT3\endcsname{\color[rgb]{1,0,1}}%
      \expandafter\def\csname LT4\endcsname{\color[rgb]{0,1,1}}%
      \expandafter\def\csname LT5\endcsname{\color[rgb]{1,1,0}}%
      \expandafter\def\csname LT6\endcsname{\color[rgb]{0,0,0}}%
      \expandafter\def\csname LT7\endcsname{\color[rgb]{1,0.3,0}}%
      \expandafter\def\csname LT8\endcsname{\color[rgb]{0.5,0.5,0.5}}%
    \else
      \def\colorrgb#1{\color{black}}%
      \def\colorgray#1{\color[gray]{#1}}%
      \expandafter\def\csname LTw\endcsname{\color{white}}%
      \expandafter\def\csname LTb\endcsname{\color{black}}%
      \expandafter\def\csname LTa\endcsname{\color{black}}%
      \expandafter\def\csname LT0\endcsname{\color{black}}%
      \expandafter\def\csname LT1\endcsname{\color{black}}%
      \expandafter\def\csname LT2\endcsname{\color{black}}%
      \expandafter\def\csname LT3\endcsname{\color{black}}%
      \expandafter\def\csname LT4\endcsname{\color{black}}%
      \expandafter\def\csname LT5\endcsname{\color{black}}%
      \expandafter\def\csname LT6\endcsname{\color{black}}%
      \expandafter\def\csname LT7\endcsname{\color{black}}%
      \expandafter\def\csname LT8\endcsname{\color{black}}%
    \fi
  \fi
  \setlength{\unitlength}{0.0500bp}%
  \begin{picture}(7200.00,5040.00)%
    \gplgaddtomacro\gplbacktext{%
      \csname LTb\endcsname%
      \put(536,755){\makebox(0,0)[r]{\strut{}\scalebox{1.35}{$0$}}}%
      \put(536,1592){\makebox(0,0)[r]{\strut{}\scalebox{1.35}{$0.1$}}}%
      \put(536,2428){\makebox(0,0)[r]{\strut{}\scalebox{1.35}{$0.2$}}}%
      \put(536,3265){\makebox(0,0)[r]{\strut{}\scalebox{1.35}{$0.3$}}}%
      \put(536,4101){\makebox(0,0)[r]{\strut{}\scalebox{1.35}{$0.4$}}}%
      \put(536,4938){\makebox(0,0)[r]{\strut{}\scalebox{1.35}{$0.5$}}}%
      \put(720,500){\makebox(0,0){\strut{}\scalebox{1.35}{$3.2\times 10^5$}}}%
      \put(3780,500){\makebox(0,0){\strut{}\scalebox{1.35}{$3.2\times 10^6$}}}%
      \put(6839,500){\makebox(0,0){\strut{}\scalebox{1.35}{$3.2\times 10^7$}}}%
      \put(-62,2846){\rotatebox{-270}{\makebox(0,0){\strut{}\scalebox{1.5}{$\mathbb{P}_e^{(w)}$ (detection error probability)}}}}%
      \put(3779,194){\makebox(0,0){\strut{}\scalebox{1.5}{$n$ (optical modes)}}}%
      \put(1774,2261){\makebox(0,0)[l]{\strut{}\textbf{Gauss}}}%
      \put(1345,2261){\makebox(0,0)[l]{\strut{}\textbf{MC}}}%
      \put(882,2261){\makebox(0,0)[l]{\strut{}\textbf{Exp}}}%
    }%
    \gplgaddtomacro\gplfronttext{%
      \csname LTb\endcsname%
      \put(2201,2007){\makebox(0,0)[l]{\strut{}\scalebox{1.1}{~~~~$\zeta=0.25\sqrt{Q/n}$}}}%
      \csname LTb\endcsname%
      \put(2201,1667){\makebox(0,0)[l]{\strut{}\scalebox{1.1}{~~~~$\zeta=0.03\sqrt[4]{Q/n}$}}}%
      \csname LTb\endcsname%
      \put(2201,1327){\makebox(0,0)[l]{\strut{}\scalebox{1.1}{~~~~$\zeta=0.003$}}}%
      \csname LTb\endcsname%
      \put(2201,987){\makebox(0,0)[l]{\strut{}\scalebox{1.1}{~~~~$\zeta=0.008$}}}%
    }%
    \gplbacktext
    \put(0,0){\includegraphics{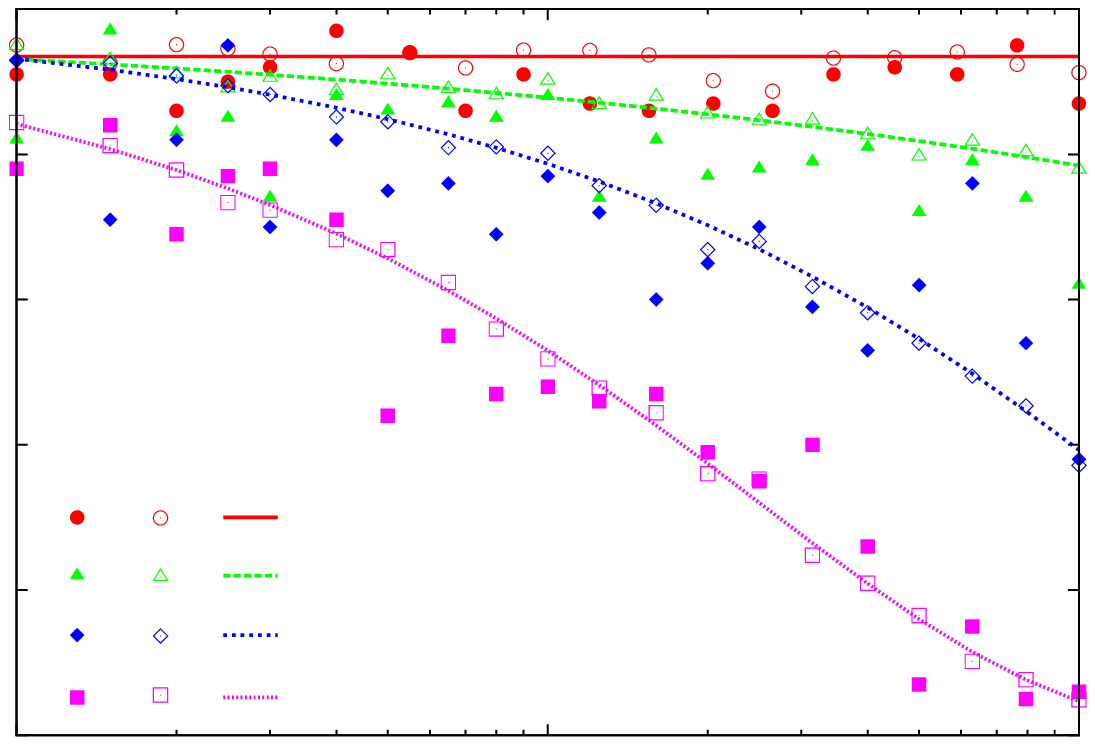}}%
    \gplfronttext
  \end{picture}%
\endgroup

%% file: conclusion.tex
\section*{Discussion}
\label{sec:conclusion}
We determined that covert communication is achievable provided that 
  the adversary's measurement is subject to non-adversarial excess noise.
Excess noise is crucial, as pure loss alone does not allow
  covert communication, starkly contrasting the QKD scenario. 
However, the existence of excess noise in practical systems  
  (e.g.,~blackbody radiation and dark counts)
  allows covert communication, as demonstrated for the first time in our
  proof-of-concept optical covert communication system.
Even though our results are for an optical channel, they are 
  relevant to RF communication due to the recent advances in 
  quantum-noise-limited microwave-frequency amplifiers and 
  detectors~\cite{abdo14josephson}.
Finally, our work provides a significant impetus towards the development of
  covert optical \emph{networks}, eventually scaling privacy to large 
  interconnected systems.

%% file: methods.tex
\section*{Methods}

\subsection*{Covert Optical Communication Theorems}

Here we state our theorems, with proofs deferred to the Supplementary 
  Information.
Each theorem can be classified as either an ``achievability'' or a ``converse''.
Achievability theorems (\ref{th:thermal}, \ref{th:dark}, and 
  \ref{th:lpd_ppm}) establish the lower limit on the amount of information
  that can be covertly transmitted from Alice to Bob, while the converse
  theorems (\ref{th:pureloss} and \ref{th:converse})
  demonstrate the upper limit.
In essence, the achievability results are obtained by 
\begin{enumerate}
\item fixing Alice's and Bob's communication system, revealing its construction 
  in entirety (except the shared secret) to Willie;
\item showing that, even with such information, any detector Willie can choose 
  within some natural constraints is ineffective at discriminating Alice's 
  transmission state; and 
\item demonstrating that the transmission can be reliably
  decoded by Bob using the shared secret.
\end{enumerate}
On the other hand, converses are established by 
\begin{enumerate}
\item fixing Willie's detection scheme (and revealing it to Alice and Bob); and
\item demonstrating that no amount of 
  resources allows Alice to both remain undetected by Willie and exceed the
  upper limit on the amount of information that is reliably transmitted to Bob.
\end{enumerate}
We start by claiming the inability to instantiate covert communication in the
  absence of excess noise.

\begin{theorem}\emph{(Insufficiency of pure-loss for covert 
  communication)}\label{th:pureloss}
\input{th_pureloss}

\end{theorem}

Next we claim the achievability of the square root law when Willie's
  channel is subject to excess noise.
We first consider a lossy optical channel with additive thermal noise,
  and claim achievability even when Willie has arbitrary resources such as any
  quantum-limited measurement on the \emph{isometric extension} of the 
  Alice-to-Bob quantum channel (i.e., access to all signaling photons not 
  captured by Bob).

\begin{theorem}\emph{(Square root law for the thermal noise
channel)}\label{th:thermal}
\input{th_thermal}
\end{theorem}

In the remaining theorems Willie's detector is a noisy photon number resolving
  (PNR) receiver. 
An ideal PNR receiver is an asymptotically optimal detector for Willie in the 
  pure-loss regime (as discussed in the remark following the proof of Theorem
  \ref{th:pureloss} in the Supplementary Information).
However, any practical implementation of a PNR receiver has a non-zero 
  dark current.
Theorems \ref{th:dark} and \ref{th:lpd_ppm} show that noise from the 
  resulting dark counts enables covert communication even over a pure-loss
  channel.
We model the dark counts per mode in Willie's PNR detector as a 
  Poisson process with average number of dark counts per mode $\lambda_w$.

\begin{theorem}\emph{(Dark counts yield square root law)}
\label{th:dark}
\input{th_dark}

\end{theorem}

The proof of Theorem \ref{th:dark} demonstrates that $\mathcal{O}(\sqrt{n})$ 
  covert bits can be reliably transmitted using be on-off keying (OOK) coherent 
  state modulation where Alice transmits the {\em on} symbol 
  $\ket{\alpha}$ with probability $q=\mathcal{O}(1/\sqrt{n})$ and the {\em off} 
  symbol $\ket{0}$ with probability $1-q$.
However, the skewed on-off duty cycle of OOK modulation makes construction of
  efficient error correction codes (ECCs) challenging. 
We thus consider \emph{pulse position modulation} (PPM) which
  constrains the OOK signaling scheme, enabling the use of many
  efficient ECCs by sacrificing a constant fraction of throughput.
Each PPM symbol uses a PPM {\em frame} to transmit a sequence of $Q$ coherent 
  state pulses, $\ket{0}\ldots\ket{\alpha}\ldots\ket{0}$, encoding message 
  $i\in\{1,2,\ldots,Q\}$ by transmitting $\ket{\alpha}$ in the 
  $i^{\text{th}}$ mode of the PPM frame. 
Next we claim that the square root scaling is achievable under this
  structural constraint.

\begin{theorem}\emph{(Dark counts yield square root law under structured modulation)}
\label{th:lpd_ppm}
\input{th_lpd_ppm}

\end{theorem}

Finally, we claim the unsurmountability of the square root law.
We assume non-zero thermal noise ($\bar{n}_T>0$) in the channel and non-zero
  dark count rate ($\lambda_w>0$) in Willie's detector.
Setting $\lambda_w=0$ yields the converse for Theorem
  \ref{th:thermal}, and setting $\bar{n}_T=0$ yields the converse 
  for Theorems \ref{th:dark} and \ref{th:lpd_ppm}
Setting $\lambda_w=0$ and $\bar{n}_T=0$ yields the conditions for Theorem 
  \ref{th:pureloss}.
To state the theorem, we use the following asymptotic notation~\cite{clrs2e}: 
  we say $f(n)=\omega(g(n))$ when $g(n)$ is a lower bound
  that is not asymptotically tight.
  
\begin{theorem}\emph{(Converse of the square root law)}\label{th:converse}
\input{th_converse}
\end{theorem}
The restriction on the photon number variance of Alice's input states is not 
  onerous since it subsumes all well-known quantum states of a bosonic mode.
However, proving this theorem for input states with unbounded photon number 
  variance per mode remains an open problem.

Next we provide details of the experimental methodology.

\subsection*{Experimental methodology}
\subsubsection*{Alice's encoder}
Prior to communication, Alice and Bob secretly select a random subset 
  $\mathcal{S}$ of PPM frames to use for transmission: each of the $n/Q$ 
  available PPM frames is selected independently with probability $\zeta$. 
Alice and Bob then secretly generate a vector $\mathbf{k}$ containing 
  $|\mathcal{S}|$ numbers selected independently uniformly at random from 
  $\{0,1,\ldots,Q-1\}$, where $|\mathcal{S}|$ denotes the cardinality of
  $\mathcal{S}$. 
Alice encodes a message into a codeword of size $|\mathcal{S}|$ using a
  Reed-Solomon (RS) code.
She adds $\mathbf{k}$ modulo $Q$ to this message and transmits it on the 
  PPM frames in $\mathcal{S}$. 
We note that this is almost identical to the construction of the coding
  scheme in the proof of Theorem \ref{th:lpd_ppm} (see the Supplementary
  Information), with the exception of the use of an RS code for error 
  correction.

\subsubsection*{Generation of transmitted symbols}
Alice generates the length-$n$ binary sequence describing the transmitted
signal, with a ``1'' at a given location indicating a pulse in that mode,
and a ``0'' indicating the absence of a pulse.  First, Alice encodes
random data, organized into $Q$-ary symbols, with
an RS code and modulo-$Q$ addition of $\mathbf{k}$ to produce a coded
  sequence of $Q$-ary symbols.
The value of the $i^{\text{th}}$ symbol in this sequence indicates which
mode in the $i^{\text{th}}$ PPM symbol in the set $\mathcal{S}$ contains a 
  pulse,
whereas all modes of the PPM frames not in $\mathcal{S}$
remain empty.  
Mapping occupied modes to ``1'' and unoccupied modes to ``0'' results in the
desired length-$n$ binary sequence.

To accurately estimate Willie's detection error probability in the face
  of optical power fluctuations, the length-$n$ binary sequence
from above is alternated with a length-$n$ sequence of all ``0''s, to
produce the final length-$2n$ sequence that is passed to the experimental
  setup.
Willie gets a ``clean'' look at the channel when Alice is silent using these 
  interleaved ``0''s, thus allowing the estimation of both the false alarm 
  and the missed detection probabilities under the same conditions.
Bob simply discards the interleaved ``0''s.

\subsubsection*{Bob's decoder}
Bob examines only the PPM frames in $\mathcal{S}$. 
If two or more pulses are detected in a PPM frame, one of them is selected 
  uniformly at random. 
If no pulses are detected, it is labeled as an \emph{erasure}.
After subtracting $\mathbf{k}$ modulo $Q$ from this vector of PPM symbols 
  (subtraction is not performed on erasures), the resultant vector is passed 
  to the RS decoder.

For each experiment we record the total number of bits in the 
  successfully-decoded codewords; the undecoded codewords are discarded.
For each pair of parameters $(\zeta,n)$ we report the mean of the total number 
  of decoded bits over 100 experiments.
The reported symbol error rate is the total number of lost data symbols 
  during all the experiments at the specified communication regime divided
  by the total number of data symbols transmitted during these experiments.
The calculation of the theoretical channel capacity is presented in the
  Supplementary Information.

\subsubsection*{Willie's detector}
\emph{Estimation of} $\mathbb{P}_e^{(w)}$---The test statistic for the 
  log-likelihood ratio test is defined as:
\begin{align}
\label{eq:pure_LRT}L&=\log\frac{f_1(\mathbf{x}_w)}{f_0(\mathbf{x}_w)}=\sum_{i=1}^{n/Q}\log\left[1+\zeta p_r^{(w)}\left(\frac{y_i^{(w)}}{Qp_D^{(w)}}-1\right)\right],
\end{align}
where $f_0(\mathbf{x}_w)$ and $f_1(\mathbf{x}_w)$ are the likelihood
  functions of the click record $\mathbf{x}_w$ corresponding to Alice being
  quiet and transmitting, $y^{(w)}_i$ is the number of clicks Willie
  observes in the $i^{\text{th}}$ PPM frame, and 
  $p^{(w)}_r=1-e^{-\eta_w\bar{n}}$ is the probability of Willie observing 
  a click stemming from Alice's transmission.
Equation \eqref{eq:pure_LRT} is derived in the Supplementary Information.
Willie calculates $L$ using equation \eqref{eq:pure_LRT} and compares it to a
  threshold $S$, accusing Alice if $L\geq S$.
Willie chooses the value of $S$ that minimizes Willie's detection error
  probability $\mathbb{P}_e^{(w)}$.

For each pair of parameters $(n,\zeta)$ as well as Alice's transmission state,
  we perform $m$ experiments, obtaining a sample vector $\mathbf{y}_w$ from
  each experiment and calculating the log-likelihood ratio $L$
  using \eqref{eq:pure_LRT}.
We denote by $\mathbf{L}^{(0)}=[L^{(0)}_1,\ldots,L^{(0)}_m]$ 
  and $\mathbf{L}^{(1)}=[L^{(1)}_1,\ldots,L^{(1)}_m]$ the vectors of
  experimentally observed log-likelihood ratios when Alice does not transmit and
  transmits, respectively.
To estimate Willie's probability of error $\mathbb{P}_e^{(w)}$, we construct
  empirical distribution functions
  $\hat{F}^{(0)}_{m}(x)=\frac{1}{n}\sum_{i=1}^m\mathbf{1}_{L^{(0)}_i\leq x}(x)$
  and
  $\hat{F}^{(1)}_{m}(x)=\frac{1}{m}\sum_{i=1}^m\mathbf{1}_{L^{(1)}_i\leq x}(x)$,
  where $\mathbf{1}_{\mathcal{A}}(x)=\left\{\begin{array}{rl}1&\text{if}~x\in~\mathcal{A}\\0&\text{if}~x\notin~\mathcal{A}\end{array}\right.$ denotes
  the indicator function.
The estimated probability of error is then 
\begin{align}
\label{eq:pe_estimate}\hat{\mathbb{P}}_e^{(w)}&=\frac{1}{2}\min_{S}(1-\hat{F}^{(0)}_{m}(S)+\hat{F}^{(1)}_{m}(S)).
\end{align}
\\
\emph{Monte-Carlo simulation and Gaussian approximation}---We perform a 
  Monte-Carlo study using $10^5$ simulations per $(n,\zeta)$ pair.
We generate, encode, and detect the messages as in the physical experiment,
  and use equation \eqref{eq:pe_estimate} to estimate Willie's probability of 
  error, but simulate the optical channel induced by our choice of a laser-light
  transmitter and an SPD using its measured characteristics reported
  in Table~\ref{tab:exp_params}.
Similarly, we use the values in Table~\ref{tab:exp_params} for our analytical
  Gaussian approximation of $\mathbb{P}_e^{(w)}$ described in the
  Supplementary Information.
\\

\emph{Confidence intervals}---We compute the confidence intervals for the 
  estimate in equation \eqref{eq:pe_estimate} using
  Dvoretzky-Keifer-Wolfowitz inequality \cite{dkw56,massart90dkw}, 
  which relates the distribution 
  function $F_X(x)$ of random variable $X$ to the empirical distribution 
  function $\hat{F}_m(x)=\frac{1}{m}\sum_{i=1}^m\mathbf{1}_{X_i\leq x}(x)$ 
  associated with a sequence $\{X_i\}_{i=1}^m$ of $m$ i.i.d.~draws of the
  random variable $X$ as follows:
\begin{align}
\mathbb{P}(\sup_{x}|\hat{F}_m(x)-F_X(x)|>\xi)&\leq2e^{-2m\xi^2},
\end{align}
  where $\xi>0$.
For $x_0$, the $(1-\alpha)$ confidence interval for the empirical
  estimate of $F(x_0)$ is given by 
  $[\max\{\hat{F}_m(x_0)-\xi,0\},\min\{\hat{F}_m(x_0)+\xi,1\}]$ where 
  $\xi=\sqrt{\frac{\log(2/\alpha)}{2m}}$.
Thus, $\pm\xi$ is used for reporting the confidence intervals in Figure 
  \ref{fig:willie}.

%% file: th_pureloss.tex
Suppose Willie has a pure-loss channel from Alice and is limited only 
  by the laws of physics in his receiver measurement choice.
Then Alice cannot communicate to Bob reliably and covertly even if Alice and Bob
  have access to a pre-shared secret of unbounded size, an unattenuated 
  observation of the transmission, and a quantum-optimal receiver.

%% file: th_thermal.tex
Suppose Willie has access to an arbitrarily complex receiver measurement 
  as permitted by the laws of quantum physics
  and can capture all the photons transmitted by Alice that do not reach Bob.
Let Willie's channel from 
  Alice be subject to noise from a thermal environment that injects 
  $\bar{n}_T>0$ photons per optical mode on average, and let Alice and Bob share
  a secret of sufficient length before communicating.
Then Alice can lower-bound Willie's detection error probability 
  $\mathbb{P}_e^{(w)}\geq\frac{1}{2}-\epsilon$ for any
  $\epsilon>0$ while reliably transmitting $\mathcal{O}(\sqrt{n})$ bits
  to Bob in $n$ optical modes even if Bob only has access to a (sub-optimal) 
  coherent detection receiver, such as an optical homodyne detector.

%% file: th_dark.tex
Suppose that Willie has a pure-loss channel from Alice, captures all
  photons transmitted by Alice that do not reach Bob, but is limited to
  a receiver with a non-zero dark current.
Let Alice and Bob share a secret of sufficient length before communicating.
Then Alice can lower-bound Willie's detection error probability 
  $\mathbb{P}_e^{(w)}\geq\frac{1}{2}-\epsilon$ for any
  $\epsilon>0$ while reliably transmitting $\mathcal{O}(\sqrt{n})$ bits
  to Bob in $n$ optical modes.

%% file: th_lpd_ppm.tex
Suppose that Willie has a pure-loss channel from Alice,
  can capture all photons transmitted by Alice that do not reach Bob,
  but is limited to a PNR receiver with a non-zero dark current.
Let Alice and Bob share a secret of sufficient length before communicating.
Then Alice can lower-bound Willie's detection error probability
  $\mathbb{P}_e^{(w)}\geq\frac{1}{2}-\epsilon$ for any
  $\epsilon>0$ while reliably transmitting 
  $\mathcal{O}(\sqrt{\frac{n}{Q}}\log Q)$ bits
  to Bob using $n$ optical modes and a $Q$-ary PPM constellation.

%% file: th_converse.tex
Suppose Alice only uses $n$-mode codewords with total photon number variance
  $\sigma^2_x=\mathcal{O}(n)$.
Then, if she attempts to transmit $\omega(\sqrt{n})$ bits in $n$
  modes, as $n\rightarrow\infty$, she is either detected by Willie with 
  arbitrarily low detection error probability, or Bob 
  cannot decode with arbitrarily low decoding error probability.

%% file: supplement.tex
\setcounter{theorem}{0}

\section{Covert Optical Communication Theorems}

Here we re-state the theorems from the main paper and provide their proofs.

\begin{theorem}\emph{(Insufficiency of pure-loss for covert 
  communication)}\label{th:pureloss_s}
\input{th_pureloss}
\end{theorem}

In the proof of this theorem we denote a tensor product of $n$ Fock (or photon
  number) states by 
 $\ket{\mathbf{u}}\equiv\ket{u_1}\otimes\ket{u_2}\otimes\cdots\otimes\ket{u_n}$,
  where vector $\mathbf{u}\in\mathbb{N}_{0}^n$ and $\mathbb{N}_{0}$ 
  is the set of non-negative integers.
Specifically, $\ket{\mathbf{0}}\equiv\ket{0}^{\otimes n}$.
Before proceeding with the proof, we prove the following lemma:
\begin{lemma}
\label{lemma:bs_on_phi}
Given the input of $n$-mode vacuum state $\ket{\mathbf{0}}^{E^n}$
  on the ``environment'' port and an $n$-mode entangled state 
  $\ket{\psi}^{A^n}=\sum_{\mathbf{k}}a_{\mathbf{k}}\ket{\mathbf{k}}^{A^n}$ 
  on the ``Alice'' port of a beamsplitter with transmissivity $\eta_b=1-\eta_w$,
  the diagonal elements of the output state $\rho^{W^n}$
  on the ``Willie'' port can be expressed in the $n$-fold Fock state basis as 
  follows:
\begin{align}
\label{eq:diag_pl_output}{\vphantom{\ket{\phi}}}^{W^n}\hspace{-4pt}\bra{\mathbf{s}}\hat{\rho}^{W^n}\ket{\mathbf{s}}^{W^n}&=\sum_{\mathbf{k}\in\mathbb{N}_{0}^n}\left|a_{\mathbf{k}}\right|^2\prod_{i=1}^n\binom{k_i}{s_i}(1-\eta_w)^{k_i-s_i}\eta_w^{s_i}.
\end{align}
\end{lemma}
\begin{proof}
A beamsplitter can be described as a unitary transformation $U_{BS}$ from 
  the two input modes (Alice's and the environment's ports) to the two output 
  modes (Bob's and Willie's ports).
Given a Fock state input $\ket{t}^A$ on Alice's port and vacuum input 
  $\ket{0}^E$ on the environment's port, the output at Bob's and Willie's ports
  is described as follows~\cite[Section IV.D]{campos89beamsplitter}:
\begin{align*}
U_{BS}\ket{t}^A\ket{0}^E&=\sum_{m=0}^t\sqrt{\binom{t}{m}\eta_w^m(1-\eta_w)^{t-m}}\ket{m}^W\ket{t-m}^B.
\end{align*}
Thus,
\begin{align*}
U_{BS}^{\otimes n}\ket{\mathbf{t}}^{A^n}\ket{\mathbf{0}}^{E^n}&=\bigotimes_{i=1}^n\sum_{m_i=0}^{t_i}\sqrt{\binom{t_i}{m_i}\eta_w^{m_i}(1-\eta_w)^{t_i-m_i}}\ket{m_i}^{W_i}\ket{t_i-m_i}^{B_i},
\end{align*}
which implies
\begin{align*}
U_{BS}^{\otimes n}\ket{\psi}^{A^n}\ket{\mathbf{0}}^{E^n}&=\sum_{\mathbf{t}\in\mathbb{N}_{0}^n}a_{\mathbf{t}}\bigotimes_{i=1}^n\sum_{m_i=0}^{t_i}\sqrt{\binom{t_i}{m_i}\eta_w^{m_i}(1-\eta_w)^{t_i-m_i}}\ket{m_i}^{W_i}\ket{t_i-m_i}^{B_i}\equiv\ket{\phi}^{W^nB^n}.
\end{align*}
Now, the partial trace of the output state $\rho^{BW}=\ket{\phi}^{W^nB^n}$ over 
  Bob's system reveals Willie's output state:
\begin{align*}
\rho^{W^n}&=\trace_{B^n}\left[\ket{\phi}^{W^nB^n}{\vphantom{\ket{\phi}}}^{W^nB^n}\hspace{-3pt}\bra{\phi}\right]\\
&=\sum_{\mathbf{x}\in\mathbb{N}_{0}^n}\left|{\vphantom{\ket{\phi}}}^{B^n}\hspace{-2pt}\langle\mathbf{x}|\phi\rangle^{W^nB^n}\right|^2,
\end{align*}
with 
\begin{align}
{\vphantom{\ket{\phi}}}^{B^n}\hspace{-2pt}\langle\mathbf{x}|\phi\rangle^{W^nB^n}&=\sum_{\mathbf{t}\in\mathbb{N}_{0}^n}a_{\mathbf{t}}\bigotimes_{i=1}^n\sum_{m_i=0}^{t_i}\sqrt{\binom{t_i}{m_i}\eta_w^{m_i}(1-\eta_w)^{t_i-m_i}}\ket{m_i}^{W_i}{\vphantom{\ket{m_i}}}^{B_i}\langle x_i|t_i-m_i\rangle^{B_i}\nonumber\\
\label{eq:orth}&=\sum_{\mathbf{t}\in\mathbb{N}_{0}^n}a_{\mathbf{t}}\bigotimes_{i=1}^n\sqrt{\binom{t_i}{x_i}\eta_w^{t_i-x_i}(1-\eta_w)^{x_i}}\ket{t_i-x_i}^{W_i},
\end{align}
where equation \eqref{eq:orth} is due to the orthogonality of the Fock states.
Thus, 
\begin{align}
\label{eq:sx}{\vphantom{\ket{\phi}}}^{W^n}\hspace{-4pt}\bra{\mathbf{s}}\hat{\rho}^{W^n}\ket{\mathbf{s}}^{W^n}&=\sum_{\mathbf{x}\in\mathbb{N}_{0}^n}\left|{\vphantom{\ket{\phi}}}^{W^n}\hspace{-4pt}\bra{\mathbf{s}}{\vphantom{\ket{\phi}}}^{B^n}\hspace{-2pt}\langle\mathbf{x}|\phi\rangle^{W^nB^n}\right|^2,
\end{align}
where
\begin{align}
{\vphantom{\ket{\phi}}}^{W^n}\hspace{-4pt}\bra{\mathbf{s}}{\vphantom{\ket{\phi}}}^{B^n}\hspace{-2pt}\langle\mathbf{x}|\phi\rangle^{W^nB^n}&=\sum_{\mathbf{t}\in\mathbb{N}_{0}^n}a_{\mathbf{t}}\prod_{i=1}^n\sqrt{\binom{t_i}{x_i}\eta_w^{t_i-x_i}(1-\eta_w)^{x_i}}\delta_{s_i,t_i-x_i}\nonumber\\
\label{eq:x+s=k}&=a_{\mathbf{x}+\mathbf{s}}\prod_{i=1}^n\sqrt{\binom{x_i+s_i}{x_i}\eta_w^{s_i}(1-\eta_w)^{x_i}},
\end{align}
with $\delta_{a,b}=\left\{\begin{array}{rl}1&\text{if~}a=b\\0&\text{otherwise}\end{array}\right.$.
Substituting $\mathbf{k}=\mathbf{x}+\mathbf{s}$ into equation
  \eqref{eq:x+s=k} and substituting the right-hand side (RHS) of 
  \eqref{eq:x+s=k} into equation \eqref{eq:sx} yields
\begin{align}
{\vphantom{\ket{\phi}}}^{W^n}\hspace{-4pt}\bra{\mathbf{s}}\hat{\rho}^{W^n}\ket{\mathbf{s}}^{W^n}&=\sum_{\mathbf{k}\in\mathbb{N}_{0}^n}\left|a_{\mathbf{k}}\prod_{i=1}^n\sqrt{\binom{k_i}{s_i}\eta_w^{s_i}(1-\eta_w)^{k_i-s_i}}\right|^2\nonumber\\
\label{eq:realprod}&=\sum_{\mathbf{k}\in\mathbb{N}_{0}^n}\left|a_{\mathbf{k}}\right|^2\prod_{i=1}^n\binom{k_i}{s_i}\eta_w^{s_i}(1-\eta_w)^{k_i-s_i}
\end{align}
where equation \eqref{eq:realprod} is due to $\eta_w\in[0,1)$.
\end{proof}
\\
\begin{proof}\emph{(Theorem \ref{th:pureloss_s})}
Alice sends one of $2^M$ (equally likely) $M$-bit messages by choosing an 
  element from an arbitrary codebook $\{\hat{\rho}^{A^n}_x,x=1,\ldots,2^M\}$,
  where a state 
  $\hat{\rho}^{A^n}_x=\ket{\psi_x}^{A^nA^n}\hspace{-4pt}\bra{\psi_x}$
  encodes an $M$-bit message $W_x$.
$\ket{\psi_x}^{A^n}=\sum_{\mathbf{k}\in\mathbb{N}_{0}^n}a_{\mathbf{k}}(x)\ket{\mathbf{k}}^{A^n}$ is a general 
  $n$-mode pure state, where
  $\ket{\mathbf{k}}\equiv\ket{k_1}\otimes\ket{k_2}\otimes\cdots\otimes\ket{k_n}$
  is a tensor product of $n$ Fock states. 
We limit our analysis to pure input states since, by convexity, using mixed
  states as inputs can only degrade the performance (since that is
  equivalent to transmitting a randomly chosen pure state from an ensemble and 
  discarding the knowledge of that choice).

Let Willie use an ideal SPD on all $n$ modes, given by positive operator-valued
  measure (POVM)
  $\left\{\ket{0}\bra{0},\sum_{j=1}^\infty\ket{j}\bra{j}\right\}^{\otimes n}$.
When $W_u$ is transmitted, Willie's hypothesis test reduces to 
  discriminating between the states
\begin{align}
\label{eq:pureloss_rho0}\hat{\rho}_0^{W^n}&=\ket{\mathbf{0}}^{W^nW^n}\hspace{-4pt}\bra{\mathbf{0}} \;{\text{and}}\\
\label{eq:pureloss_rho1}\hat{\rho}_1^{W^n}&=\hat{\rho}^{W^n}_u,
\end{align}
where $\hat{\rho}^{W^n}_u$ is the output state of a pure-loss channel with 
  transmissivity $\eta_w$ corresponding to an input state $\hat{\rho}^{A^n}_u$.
Thus, Willie's average error probability is:
\begin{align}
\label{eq:pew}\mathbb{P}_e^{(w)}&=\frac{1}{2^{M+1}}\sum_{u=1}^{2^M}{\vphantom{\ket{\phi}}}^{W^n}\hspace{-4pt}\bra{\mathbf{0}}\hat{\rho}^{W^n}_u\ket{\mathbf{0}}^{W^n},
\end{align}
  since messages are sent equiprobably.
Note that the error is entirely due to missed codeword detections,
  as Willie's receiver never raises a false alarm.
By Lemma \ref{lemma:bs_on_phi},
\begin{align}
{\vphantom{\ket{\phi}}}^{W^n}\hspace{-4pt}\bra{\mathbf{0}}\hat{\rho}^W_u\ket{\mathbf{0}}^{W^n}&=\sum_{\mathbf{k}\in\mathbb{N}_{0}^n}\left|a_{\mathbf{k}}(u)\right|^2(1-\eta_w)^{\sum_{i=1}^n k_i}\nonumber\\
&\leq\left|a_{\mathbf{0}}(u)\right|^2 + (1-\left|a_{\mathbf{0}}(u)\right|^2)(1-\eta_w)\nonumber\\
\label{eq:pgf_ub}&=1-\eta_w\left(1-\left|a_{\mathbf{0}}(u)\right|^2\right).
\end{align}
Substituting equation \eqref{eq:pgf_ub} into equation \eqref{eq:pew} yields:
\begin{align*}
\mathbb{P}_e^{(w)}&\leq\frac{1}{2}-\frac{\eta_w}{2}\left(1-\frac{1}{2^M}\sum_{u=1}^{2^M}\left|a_{\mathbf{0}}(u)\right|^2\right).
\end{align*}
Thus, to ensure $\mathbb{P}_e^{(w)}\geq\frac{1}{2}-\epsilon$, Alice must use a 
  codebook with the probability of transmitting zero photons:
\begin{align}
\label{eq:restrict_a0}\frac{1}{2^M}\sum_{u=1}^{2^M}\left|a_{\mathbf{0}}(u)\right|^2&\geq1-\frac{2\epsilon}{\eta_w}.
\end{align}
Equation \eqref{eq:restrict_a0} can be restated as an upper
  bound on the probability of transmitting one or more photons:
\begin{align}
\label{eq:restrict_not_a0}\frac{1}{2^M}\sum_{u=1}^{2^M}\left(1-\left|a_{\mathbf{0}}(u)\right|^2\right)&\leq\frac{2\epsilon}{\eta_w}.
\end{align}
Now we show that there exists an interval $(0,\epsilon_0]$, $\epsilon_0>0$
  such that if $\epsilon\in(0,\epsilon_0]$, Bob's average decoding error
  probability $\mathbb{P}_e^{(b)}\geq\delta_0$ where $\delta_0>0$, thus making 
  covert communication over a pure-loss channel unreliable.

Denote by $E_{u\rightarrow v}$ the event that the transmitted message
  $W_u$ is decoded by Bob as $W_v\neq W_u$.
Given that $W_u$ is transmitted, the decoding error probability 
  is the probability of the union of events 
  $\cup_{v=0,v\neq u}^{2^M}E_{u\rightarrow v}$.
Let Bob choose a POVM $\{\Lambda_j^*\}$ that minimizes the average probability
  of error over $n$ optical channel modes:
\begin{align}
\label{eq:pe_b_1}\mathbb{P}_e^{(b)}&=\inf_{\{\Lambda_j\}}\frac{1}{2^M}\sum_{u=1}^{2^M}\mathbb{P}\left(\cup_{v=0,v\neq u}^{2^M}E_{u\rightarrow v}\right).
\end{align}
Now consider a codebook that meets the necessary condition for covert 
  communication given in equation \eqref{eq:restrict_not_a0}.
Define the subset of this codebook 
  $\left\{\hat{\rho}^{A^n}_u,u\in\mathcal{A}\right\}$ where
  $\mathcal{A}=\left\{u:1-|a_{\mathbf{0}}(u)|^2\leq \frac{4\epsilon}{\eta_w}\right\}$.
We lower-bound \eqref{eq:pe_b_1} as follows:
\begin{align}
\label{eq:pe_b_2}\mathbb{P}_e^{(b)}&=\frac{1}{2^M}\sum_{u\in\bar{\mathcal{A}}}\mathbb{P}\left(\cup_{v=0,v\neq u}^{2^M}E_{u\rightarrow v}\right)+\frac{1}{2^M}\sum_{u\in\mathcal{A}}\mathbb{P}\left(\cup_{v=0,v\neq u}^{2^M}E_{u\rightarrow v}\right)\\
\label{eq:pe_b_3}&\geq\frac{1}{2^M}\sum_{u\in\mathcal{A}}\mathbb{P}\left(\cup_{v=0,v\neq u}^{2^M}E_{u\rightarrow v}\right),
\end{align}
where the probabilities in equation \eqref{eq:pe_b_2} are with respect to the 
  POVM $\{\Lambda_j^*\}$ that minimizes equation \eqref{eq:pe_b_1} over the 
  entire codebook.
Without loss of generality, let's assume that $|\mathcal{A}|$ is even, and 
  split $\mathcal{A}$ into two equal-sized subsets $\mathcal{A}^{(\text{left})}$
  and $\mathcal{A}^{(\text{right})}$ (formally, 
  $\mathcal{A}^{(\text{left})}\cup\mathcal{A}^{(\text{right})}=\mathcal{A}$, 
  $\mathcal{A}^{(\text{left})}\cap\mathcal{A}^{(\text{right})}=\emptyset$, 
  and $|\mathcal{A}^{(\text{left})}|=|\mathcal{A}^{(\text{right})}|$).
Let $g:\mathcal{A}^{(\text{left})}\rightarrow\mathcal{A}^{(\text{right})}$ be
  a bijection.
We can thus re-write \eqref{eq:pe_b_3}:
\begin{align}
\mathbb{P}_e^{(b)}&\geq\frac{1}{2^M}\sum_{u\in\mathcal{A}^{(\text{left})}}2\left(\frac{\mathbb{P}\left(\cup_{v=0,v\neq u}^{2^M}E_{u\rightarrow v}\right)}{2}+\frac{\mathbb{P}\left(\cup_{v=0,v\neq g(u)}^{2^M}E_{g(u)\rightarrow v}\right)}{2}\right)\nonumber\\
\label{eq:pe_b_lb}&\geq\frac{1}{2^M}\sum_{u\in\mathcal{A}^{(\text{left})}}2\left(\frac{\mathbb{P}\left(E_{u\rightarrow g(u)}\right)}{2}+\frac{\mathbb{P}\left(E_{g(u)\rightarrow u}\right)}{2}\right),
\end{align}
where the second lower bound is due to the events
  $E_{u\rightarrow g(u)}$ and $E_{g(u)\rightarrow u}$ being contained
  in the unions $\cup_{v=0,v\neq u}^{2^M}E_{u\rightarrow v}$ and 
  $\cup_{v=0,v\neq g(u)}^{2^M}E_{g(u)\rightarrow v}$, respectively.
The summation term in equation \eqref{eq:pe_b_lb},
\begin{align}
\label{eq:p_e_u}\mathbb{P}_e(u)\equiv\frac{\mathbb{P}\left(E_{u\rightarrow g(u)}\right)}{2}+\frac{\mathbb{P}\left(E_{g(u)\rightarrow u}\right)}{2},
\end{align}
  is Bob's average probability of error when Alice only sends 
  messages $W_u$ and $W_{g(u)}$ equiprobably.
We thus reduce the analytically intractable problem of
  discriminating between many states in equation \eqref{eq:pe_b_1} to a 
  quantum binary hypothesis test.

The lower bound on the probability of error in discriminating two received 
  codewords is obtained by lower-bounding the probability of error in 
  discriminating two codewords \emph{before} they are sent (this is equivalent 
  to Bob having an unattenuated unity-transmissivity channel from Alice).
Recalling that 
  $\hat{\rho}^{A^n}_u=\ket{\psi_u}^{A^nA^n}\hspace{-4pt}\bra{\psi_u}$ and
  $\hat{\rho}^{A^n}_{g(u)}=\ket{\psi_{g(u)}}^{A^nA^n}\hspace{-4pt}\bra{\psi_{g(u)}}$
  are pure states, the lower bound on the probability of error in
  discriminating between $\ket{\psi_u^{A^n}}$ and $\ket{\psi_{g(u)}^{A^n}}$
  is~\cite[Chapter IV.2 (c), Equation (2.34)]{helstrom76quantumdetect}:
\begin{align}
\label{eq:p_e_bintest}\mathbb{P}_e(u)&\geq\left.\left[1-\sqrt{1-F\left(\ket{\psi_u}^{A^n},\ket{\psi_{g(u)}}^{A^n}\right)}\right]\middle/2\right.,
\end{align}
  where  $F(\ket{\psi},\ket{\phi})=|\left<\psi|\phi\right>|^2$  is the fidelity
  between the pure states $\ket{\psi}$ and $\ket{\phi}$. 
Lower-bounding $F\left(\ket{\psi_u}^{A^n},\ket{\psi_{g(u)}}^{A^n}\right)$
  lower-bounds the RHS of equation \eqref{eq:p_e_bintest}.
For pure states $\ket{\psi}$ and $\ket{\phi}$, 
  $F(\ket{\psi},\ket{\phi})=1-\left(\frac{1}{2}\|\ket{\psi}\bra{\psi}-\ket{\phi}\bra{\phi}\|_1\right)^2$,
  where $\|\rho-\sigma\|_1$ is the trace 
  distance~\cite[Equation (9.134)]{wilde13quantumit}.
Thus,
\begin{align}
F\left(\ket{\psi_u}^{A^n},\ket{\psi_{g(u)}}^{A^n}\right)&=1-\left(\frac{1}{2}\|\hat{\rho}^{A^n}_u-\hat{\rho}^{A^n}_{g(u)}\|_1\right)^2\nonumber\\
&\geq1-\left(\frac{\|\hat{\rho}^{A^n}_u-\ket{\mathbf{0}}^{A^nA^n}\hspace{-4pt}\bra{\mathbf{0}}\|_1}{2}+\frac{\|\hat{\rho}^{A^n}_{g(u)}-\ket{\mathbf{0}}^{A^nA^n}\hspace{-4pt}\bra{\mathbf{0}}\|_1}{2}\right)^2\nonumber\\
\label{eq:back_to_F}&=1-\left(\sqrt{1-\left|{\vphantom{\ket{\phi}}}^{A^n}\hspace{-4pt}\left<\mathbf{0}|\psi_u\right>^{A^n}\right|^2}+\sqrt{1-\left|{\vphantom{\ket{\phi}}}^{A^n}\hspace{-4pt}\left<\mathbf{0}|\psi_{g(u)}\right>^{A^n}\right|^2}\right)^2,
\end{align}
  where the inequality is due to the triangle inequality 
  for trace distance.
Substituting \eqref{eq:back_to_F} into \eqref{eq:p_e_bintest} yields:
\begin{align}
\label{eq:p_e_bintest2}\mathbb{P}_e(u)&\geq\left.\left[1-\sqrt{1-\left|{\vphantom{\ket{\phi}}}^{A^n}\hspace{-4pt}\left<\mathbf{0}|\psi_u\right>^{A^n}\right|^2}-\sqrt{1-\left|{\vphantom{\ket{\phi}}}^{A^n}\hspace{-4pt}\left<\mathbf{0}|\psi_{g(u)}\right>^{A^n}\right|^2}\right]\middle/2\right..
\end{align}
Since $\left|{\vphantom{\ket{\phi}}}^{A^n}\hspace{-4pt}\left<\mathbf{0}|\psi_u\right>^{A^n}\right|^2=|a_{\mathbf{0}}(u)|^2$ and,
  by the construction of $\mathcal{A}$, 
  $1-|a_{\mathbf{0}}(u)|^2\leq\frac{4\epsilon}{\eta_w}$ and
  $1-|a_{\mathbf{0}}(g(u))|^2\leq\frac{4\epsilon}{\eta_w}$, we have:
\begin{align}
\label{eq:p_e_bintest3}\mathbb{P}_e(u)&\geq\frac{1}{2}-2\sqrt{\frac{\epsilon}{\eta_w}}.
\end{align}
Recalling the definition of $\mathbb{P}_e(u)$ in equation \eqref{eq:p_e_u}, 
  we substitute \eqref{eq:p_e_bintest3} into \eqref{eq:pe_b_lb} to obtain:
\begin{align}
\label{eq:p_e_lb2}\mathbb{P}_e^{(b)}&\geq\frac{|\mathcal{A}|}{2^M}\left(\frac{1}{2}-2\sqrt{\frac{\epsilon}{\eta_w}}\right),
\end{align}
Now, re-stating the condition for covert communication 
  \eqref{eq:restrict_not_a0} yields:
\begin{align}
\frac{2\epsilon}{\eta_w}&\geq\frac{1}{2^M}\sum_{u\in\overline{\mathcal{A}}}\left(1-\left|a_{\mathbf{0}}(u)\right|^2\right)\nonumber\\
\label{eq:frac_lb1}&\geq\frac{\left(2^M-|\mathcal{A}|\right)}{2^M}\frac{4\epsilon}{\eta_w}
\end{align}
  with equality \eqref{eq:frac_lb1} due to 
  $1-\left|a_{\mathbf{0}}(u)\right|^2> \frac{4\epsilon}{\eta_w}$ 
  for all codewords in $\overline{\mathcal{A}}$ by the construction of 
  $\mathcal{A}$.
Solving inequality in \eqref{eq:frac_lb1} for $\frac{|\mathcal{A}|}{2^M}$ yields
  the lower bound on the fraction of the codewords in $\mathcal{A}$,
\begin{align}
\label{eq:frac_lb2}\frac{|\mathcal{A}|}{2^M}&\geq\frac{1}{2}.
\end{align}
Combining equations \eqref{eq:p_e_lb2} and \eqref{eq:frac_lb2} results in a 
  positive lower bound on Bob's probability of decoding
  error 
  $\mathbb{P}_e^{(b)}\geq\frac{1}{4}-\sqrt{\frac{\epsilon}{\eta_w}}$
  for $\epsilon\in\left(0,\frac{\eta_w}{16}\right]$ and any $n$,
  and demonstrates that reliable covert communication over a pure-loss channel
  is impossible.
\end{proof}
\\

\noindent\emph{Remark}---The minimum probability of discrimination error 
  between the states given by equations \eqref{eq:pureloss_rho0} and 
  \eqref{eq:pureloss_rho1} satisfies~\cite[Section III]{pirandola08bounds}:
\begin{align*}
\frac{1-\sqrt{1-{\vphantom{\ket{\phi}}}^{W^n}\hspace{-4pt}\bra{\mathbf{0}}\hat{\rho}^{W^n}_u\ket{\mathbf{0}}^{W^n}}}{2}\leq\min\mathbb{P}^{(w)}_e\leq \frac{1}{2}{\vphantom{\ket{\phi}}}^{W^n}\hspace{-4pt}\bra{\mathbf{0}}\hat{\rho}^{W^n}_u\ket{\mathbf{0}}^{W^n}.
\end{align*}
Since $\frac{{\vphantom{\ket{\phi}}}^{W^n}\hspace{-2pt}\bra{\mathbf{0}}\hat{\rho}^{W^n}_u\ket{\mathbf{0}}^{W^n}}{4}\leq \frac{1-\sqrt{1-{\vphantom{\ket{\phi}}}^{W^n}\hspace{-2pt}\bra{\mathbf{0}}\hat{\rho}^{W^n}_u\ket{\mathbf{0}}^{W^n}}}{2}$, 
  the error probability for the SPD is at most twice that of an optimal 
  discriminator.
Thus, the SPD is an asymptotically optimal detector when the channel from
  Alice is pure-loss.
Since the photon number resolving (PNR) receiver, given by the POVM elements
  $\left\{\ket{0}\bra{0},\ket{1}\bra{1},\ket{2}\bra{2},\ldots\right\}^{\otimes n}$,
  could be used to mimic the SPD with the detection event threshold set at one
  photon, the PNR receiver is also asymptotically optimal in this scenario.
\\
\begin{theorem}\emph{(Square root law for the thermal noise
channel)}\label{th:thermal_s}
\input{th_thermal}
\end{theorem}

\noindent First, we define quantum relative entropy and prove a lemma:
\begin{definition}\emph{\textbf{Quantum relative entropy}} between states 
  $\hat{\rho}_0$ and $\hat{\rho}_1$ is 
  $D(\hat{\rho}_0\|\hat{\rho}_1)\equiv\trace\{\hat{\rho}_0(\ln\hat{\rho}_0-\ln\hat{\rho}_1)\}$.
\end{definition}

\begin{lemma}[Quantum relative entropy between two thermal states]
\label{lemma:qre_thermal}
\begingroup
\setlength{\thinmuskip}{1mu}
\setlength{\thickmuskip}{2mu}
If $\hat{\rho}_0=\sum_{n=0}^\infty \frac{{\bar{n}_0}^n}{(1+\bar{n}_0)^{1+n}}\ket{n}\bra{n}$
and $\hat{\rho}_1=\sum_{n=0}^\infty \frac{{\bar{n}_1}^n}{(1+\bar{n}_1)^{1+n}}\ket{n}\bra{n}$,
then 
$D(\hat{\rho}_0\|\hat{\rho}_1)=\bar{n}_0\ln\frac{\bar{n}_0(1+\bar{n}_1)}{\bar{n}_1(1+\bar{n}_0)}+\ln\frac{1+\bar{n}_1}{1+\bar{n}_0}$
\endgroup
\end{lemma}
\begin{proof}
Express
  $D(\hat{\rho}_0\|\hat{\rho}_1)=-\trace\{\hat{\rho}_0(\ln\hat{\rho}_1)\}-S(\hat{\rho}_0)$, 
  where $S(\hat{\rho}_0)\equiv-\trace[\hat{\rho}_0\ln\hat{\rho}_0]$ 
  is the von Neumann entropy of the state $\hat{\rho}_0$:
\begin{align}
\label{eq:vonNeumann_thermal}
S(\hat{\rho}_0)&=\ln(1+\bar{n}_0)+\bar{n}_0\ln\left(1+\frac{1}{\bar{n}_0}\right).
\end{align}
Now, 
\begin{align}
\trace[\hat{\rho}_0\ln\hat{\rho}_1]&=\trace\left[\left(\sum_{n=0}^\infty \frac{\bar{n}_0^n}{(1+\bar{n}_0)^{1+n}}\ket{n}\bra{n}\right)\left(\sum_{n=0}^\infty \ln\frac{\bar{n}_1^n}{(1+\bar{n}_1)^{1+n}}\ket{n}\bra{n}\right)\right]\nonumber\\
&=\sum_{n=0}^\infty \frac{\bar{n}_0^n}{(1+\bar{n}_0)^{1+n}}\ln\frac{\bar{n}_1^n}{(1+\bar{n}_1)^{1+n}}\nonumber\\
&=\frac{1}{1+\bar{n}_0}\ln\frac{1}{1+\bar{n}_1}\sum_{n=0}^\infty \left(\frac{\bar{n}_0}{1+\bar{n}_0}\right)^n+\ln\frac{\bar{n}_1}{1+\bar{n}_1}\sum_{n=0}^\infty \frac{n}{1+\bar{n}_0}\cdot\left(\frac{\bar{n}_0}{1+\bar{n}_0}\right)^n\nonumber\\
\label{eq:geom}&=\ln\frac{1}{1+\bar{n}_1}+\bar{n}_0\ln\frac{\bar{n}_1}{1+\bar{n}_1}
\end{align}
where \eqref{eq:geom} is due to the geometric series 
  $\sum_{n=0}^\infty \left(\frac{\bar{n}_0}{1+\bar{n}_0}\right)^n=\left(1-\frac{\bar{n}_0}{1+\bar{n}_0}\right)^{-1}$
  and
  $\sum_{n=0}^\infty \frac{n}{1+\bar{n}_0} \left(\frac{\bar{n}_0}{1+\bar{n}_0}\right)^n=\bar{n}_0$
  being the expression for the mean of geometrically-distributed random
  variable $X\sim\text{Geom}\left(\frac{1}{1+\bar{n}_0}\right)$.
Combining \eqref{eq:vonNeumann_thermal} and \eqref{eq:geom}
  yields the lemma.
\end{proof}
\\
\begin{proof}\emph{(Theorem \ref{th:thermal_s})}
\textit{Construction:} 
Let Alice use a zero-mean isotropic Gaussian-distributed coherent state input
  $\left\{p(\alpha),\ket{\alpha}\right\}$, where $\alpha \in {\mathbb C}$, 
  $p(\alpha) = e^{-|\alpha|^2/{\bar n}}/{\pi {\bar n}}$ 
  with mean photon number per symbol
  $\bar{n}=\int_{\mathbb C}|\alpha|^2 p(\alpha){\rm d}^2\alpha$. 
Alice encodes $M$-bit blocks of input into codewords of length
  $n$ symbols by generating $2^{M}$ codewords 
  $\{\bigotimes_{i=1}^n\ket{\alpha_i}_k\}_{k=1}^{2^{M}}$, each
  according to $p(\bigotimes_{i=1}^n\ket{\alpha_i})=\prod_{i=1}^{n}p(\alpha_i)$,
  where $\bigotimes_{i=1}^n\ket{\alpha_i}=\ket{\alpha_1\ldots\alpha_{n}}$ is an
  $n$-mode tensor-product coherent state.
The codebook is used only once to send a single message and is kept secret from 
  Willie, though he knows how it is constructed.

\textit{Analysis (Willie):}
Since Willie does not have access to Alice's codebook, Willie has to 
  discriminate between the following $n$-copy quantum states:
\begin{align*}
\hat{\rho}_0^{\otimes n}&=\left(\sum_{i=0}^\infty \frac{(\eta_b \bar{n}_T)^i}{(1+\eta_b \bar{n}_T)^{1+i}}\ket{i}\bra{i}\right)^{\otimes n}, \;{\text{and}}\\
\hat{\rho}_1^{\otimes n}&=\left(\sum_{i=0}^\infty \frac{(\eta_w\bar{n}+\eta_b \bar{n}_T)^i}{(1+\eta_w\bar{n}+\eta^{(n)} \bar{n}_T)^{1+i}}\ket{i}\bra{i}\right)^{\otimes n}.
\end{align*}
Willie's average probability of error in discriminating between
  $\hat{\rho}_0^{\otimes n}$ and $\hat{\rho}_1^{\otimes n}$ 
  is~\cite[Section 9.1.4]{wilde13quantumit}:
\begin{align*}
\mathbb{P}_{e}^{(w)}&\geq\frac12\left[1 - \frac12\| \hat{\rho}_1^{\otimes n} - \hat{\rho}_0^{\otimes n} \|_1\right],
\end{align*}
  where the minimum in this case is attained by a PNR detection.
The trace distance $\|\hat{\rho}_0-\hat{\rho}_1\|_1$ between states 
  $\hat{\rho}_1$ and $\hat{\rho}_1$ is upper-bounded the quantum relative 
  entropy (QRE) using quantum Pinsker's 
  Inequality~\cite[Theorem 11.9.2]{wilde13quantumit} as follows:
\begin{align*}
\|\hat{\rho}_0-\hat{\rho}_1\|_1&\leq\sqrt{2D(\hat{\rho}_0\|\hat{\rho}_1)},
\end{align*}
Thus,
\begin{align}
\label{eq:qre_pe_lb}\mathbb{P}_{e}^{(w)}&\geq\frac{1}{2}-\sqrt{\frac{1}{8}D(\hat{\rho}_0^{\otimes n}\|\hat{\rho}_1^{\otimes n})}.
\end{align}
QRE is additive for tensor product states:
\begin{align}
\label{eq:qre_additive}D(\hat{\rho}_0^{\otimes n}\|\hat{\rho}_1^{\otimes n})&=nD(\hat{\rho}_0\|\hat{\rho}_1). 
\end{align}
By Lemma \ref{lemma:qre_thermal},
\begin{align}
\label{eq:kl_therm}D(\hat{\rho}_0\|\hat{\rho}_1)&=\eta_b \bar{n}_T\ln\frac{(1+\eta_w\bar{n}+\eta_b \bar{n}_T)\eta_b \bar{n}_T}{(\eta_w\bar{n}+\eta_b \bar{n}_T)(1+\eta_b \bar{n}_T)}+\ln\frac{1+\eta_w\bar{n}+\eta_b \bar{n}_T}{1+\eta_b \bar{n}_T}.
\end{align}
The first two terms of the Taylor series expansion of the RHS of
  \eqref{eq:kl_therm} with respect to $\bar{n}$ at $\bar{n}=0$ are zero and 
  the fourth term is negative.
Thus, using Taylor's Theorem with the remainder, we can upper-bound 
  equation \eqref{eq:kl_therm} by the third term as follows:
\begin{align}
\label{eq:qre_ub}D(\hat{\rho}_0\|\hat{\rho}_1)&\leq\frac{\eta_w^2\bar{n}^2}{2\eta_b \bar{n}_T (1+\eta_b \bar{n}_T)}.
\end{align}
Combining equations \eqref{eq:qre_pe_lb}, \eqref{eq:qre_additive}, and 
  \eqref{eq:qre_ub} yields:
\begin{align}
\label{eq:thermal_pe_lb}\mathbb{P}_e^{(w)}&\geq\frac{1}{2}-\frac{\eta_w\bar{n}\sqrt{n}}{4\sqrt{\eta_b\bar{n}_T(1+\eta_b\bar{n}_T)}}
\end{align}
Therefore, setting 
\begin{align}
\label{eq:nbar}\bar{n}&=\frac{4\epsilon\sqrt{\eta_b \bar{n}_T(1+\eta_b \bar{n}_T)}}{\sqrt{n}\eta_w}
\end{align}
  ensures that Willie's error probability is lower-bounded by
  $\mathbb{P}_{e}^{(w)}\geq\frac{1}{2}-\epsilon$ over $n$ optical modes.

\textit{Analysis (Bob):} 
Suppose Bob uses a coherent detection receiver. 
A homodyne receiver, which is more efficient than a heterodyne receiver in the 
  low photon number regime~\cite{giovannetti04cappureloss}, induces an AWGN
  channel with noise power 
  $\sigma_{b}^2=\frac{2(1-\eta_b)\bar{n}_T+1}{4\eta_b}$.
Since Alice uses Gaussian modulation with symbol power $\bar{n}$ defined in 
  equation \eqref{eq:nbar}, we can upper-bound $\mathbb{P}_e^{(b)}$ by
  \cite[Equation (9)]{bash13squarerootjsac}:
\begin{align}
\label{eq:P_e_hom}\mathbb{P}_{e}^{(b)}&\leq2^{B-\frac{n}{2}\log_2\left(1+\bar{n}/2\sigma_{b}^2\right)},
\end{align}
where $B$ is the number of transmitted bits.
Substitution of $\bar{n}$ from \eqref{eq:nbar} into 
  \eqref{eq:P_e_hom} shows that $\mathcal{O}(\sqrt{n})$ bits can be
  covertly transmitted from Alice to Bob with 
  $\mathbb{P}_{e}^{(b)}<\delta$ for arbitrary $\delta>0$ given large enough
  $n$.
\end{proof}
\\

\noindent Before proving Theorems \ref{th:dark_s} and \ref{th:lpd_ppm_s}, we 
  state a lemma that is used in their proofs.

\begin{lemma}[Classical relative entropy bound on $\mathbb{P}_e$ of binary hypothesis test]
\label{lemma:cre_lb}
Denote by $\mathbb{P}_0$ and $\mathbb{P}_1$ the respective probability 
  distributions of observations when $H_0$ and $H_1$ is true.
Assuming equal prior probabilities for each hypothesis, the probability
  of discrimination error is $\mathbb{P}_e\leq\frac{1}{2}-\sqrt{\frac{1}{8}D(\mathbb{P}_0\|\mathbb{P}_1)}$, where
  $D(\mathbb{P}_0\|\mathbb{P}_1)=-\sum_{x}p_0(x)\ln\frac{p_1(x)}{p_0(x)}$ 
  is the classical relative entropy between $\mathbb{P}_0$ and $\mathbb{P}_1$
  and $p_0(x)$ and $p_1(x)$ are the respective probability mass functions
  of $\mathbb{P}_0$ and $\mathbb{P}_1$.
\end{lemma}
\begin{proof}
The minimum probability of discrimination error between $H_0$ and $H_1$ is 
  characterized by~\cite[Theorem 13.1.1]{lehmann05stathyp}:
\begin{align*}
\min\mathbb{P}_e&=\frac{1}{2}-\frac{1}{4}\|p_0(x)-p_1(x)\|_1,
\end{align*}
  where $\|a-b\|_1$ is the $\mathcal{L}_1$ norm.
By classical Pinsker's inequality~\cite[Lemma 11.6.1]{cover02IT}, 
\begin{align*}
\|p_0(x)-p_1(x)\|_1&\leq\sqrt{2D(\mathbb{P}_0\|\mathbb{P}_1)},
\end{align*}
and the lemma follows.
\end{proof}

\begin{theorem}\emph{(Dark counts yield square root law)}
\label{th:dark_s}
\input{th_dark}
\end{theorem}
\begin{proof}
\textit{Construction:} 
Let Alice use a coherent state on-off keying (OOK) modulation 
  $\left\{\pi_i, |\psi_i\rangle\langle\psi_i |\right\}$, $i = 1, 2$, where 
  $\pi_1 = 1-q$, $\pi_2 = q$, $|\psi_1\rangle = |0\rangle$, 
  $|\psi_2\rangle = |\alpha\rangle$. 
Alice and Bob generate a random codebook with each codeword symbol chosen 
  i.i.d.~from the above binary OOK constellation.

\textit{Analysis (Willie):}
Willie records vector $\mathbf{y}_w=[y_1,\ldots,y_n]$, where $y_i$ is the
  number of photons observed in the $i^{\text{th}}$ mode.
Denote by $\mathbb{P}_0$ the distribution of $\mathbf{y}_w$ when Alice
  does not transmit and by $\mathbb{P}_1$ the distribution when she transmits.
When Alice does not transmit, Willie's receiver observes a Poisson dark count
  process with rate $\lambda_w$ photons per mode.
Thus, $\{y_i\}$ is independent and identically distributed (i.i.d.) sequence of
  Poisson random variables with rate $\lambda_w$, and 
  $\mathbb{P}_0=\mathbb{P}_w^n$ where $\mathbb{P}_w=\text{Poisson}(\lambda_w)$.
When Alice transmits, although Willie captures all of her transmitted energy
  that does not reach Bob, he does not have access to Alice's and Bob's 
  codebook.
Since the dark counts are independent of the transmitted pulses,
  each observation is a mixture of two independent Poisson random variables.
Thus, each $y_i\sim\mathbb{P}_s$ is i.i.d., with 
  $\mathbb{P}_s=(1-q)\text{Poisson}(\lambda_w)+q\text{Poisson}(\lambda_w+\eta_w|\alpha|^2)$ and $\mathbb{P}_1=\mathbb{P}_s^n$.
By Lemma \ref{lemma:cre_lb},
  $\mathbb{P}_e^{(w)}\geq\frac{1}{2}-\sqrt{\frac{1}{8}D(\mathbb{P}_0\|\mathbb{P}_1)}$.
Since the classical relative entropy is additive for product distributions,
  $D(\mathbb{P}_0\|\mathbb{P}_1)=nD(\mathbb{P}_w\|\mathbb{P}_s)$.
Now,
\begin{align}
\label{eq:dark_kl}D(\mathbb{P}_w\|\mathbb{P}_s)&=-\sum_{y=0}^\infty\frac{\lambda_w^ye^{-\lambda_w}}{y!}\log\left[1-q+q\left(1+\frac{\eta_w|\alpha|^2}{\lambda_w}\right)e^{-\eta_w|\alpha|^2}\right]\\
&\leq\frac{q^2\left(e^{(\eta_w|\alpha|^2)^2/\lambda_w}-1\right)}{2}\nonumber
\end{align}
where the inequality is due to the Taylor's 
  Theorem with the remainder applied to the Taylor series expansion of 
  equation \eqref{eq:dark_kl} with respect to $q$ at $q=0$.
Thus,
\begin{align}
\label{eq:pe_dark}\mathbb{P}_e^{(w)}&\geq\frac{1}{2}-\frac{q}{4}\sqrt{n\left(e^{(\eta_w|\alpha|^2)^2/\lambda_w}-1\right)}.
\end{align}
Therefore, to ensure that $\mathbb{P}_e^{(w)}\geq\frac{1}{2}-\epsilon$, Alice 
  sets
\begin{align}
q&=\frac{4\epsilon}{\sqrt{n\left(e^{(\eta_w|\alpha|^2)^2/\lambda_w}-1\right)}}.
\end{align}

\textit{Analysis (Bob):}
Suppose Bob uses a practical single photon detector (SPD) receiver with
  probability of a dark click per mode $p_D^{(b)}$.
This induces a binary asymmetric channel between Alice and Bob,
  where the click probabilities, conditional on the input, are
  $\mathbb{P}(\text{click~}|\text{~input~}\ket{0})=p_D^{(b)}$
  and 
  $\mathbb{P}(\text{click~}|\text{~input~}\ket{\alpha})=1-e^{-\eta_b|\alpha|^2}(1-p_D^{(b)})$,
  with the corresponding no-click probabilities
  $\mathbb{P}(\text{no-click~}|\text{~input~}\ket{0})=1-p_D^{(b)}$
  and 
  $\mathbb{P}(\text{no-click~}|\text{~input~}\ket{\alpha})=e^{-\eta_b|\alpha|^2}(1-p_D^{(b)})$.
At each mode, a click corresponds to ``1'' and no-click to ``0''.
Let Bob use a maximum likelihood decoder on this sequence.
Then the standard upper bound on Bob's average decoding error
  probability is~\footnote{We use \cite[Theorem 5.6.2]{gallager68IT}, setting
  parameter $s=1$.} $\mathbb{P}_e^{(b)}\leq e^{B-nE_{0}}$,
  where $B$ is the number of transmitted bits, and $E_{0}$ is:
\begin{align*}
E_0&=-\ln\left[\left(1-p_d^{(b)}\right)\left(1-q\left(1-e^{-\frac{\eta|\alpha|^2}{2}}\right)\right)^{2}+\left((1-q)\sqrt{p_D^{(b)}}+q\sqrt{1-\left(1-p_D^{(b)}\right)e^{-\eta|\alpha|^2}}\right)^{2}\right]
\end{align*}
The Taylor series expansion of $E_0$ with respect to $q$ at $q=0$ yields
  $E_0=qC+\mathcal{O}(q^2)$, where
\begin{align*}
C&=2e^{-\eta_n|\alpha|^2/2}\left(e^{\eta_n|\alpha|^2/2}-1+p_D^{(b)}-\sqrt{p_D^{(b)}\left(e^{\eta_n|\alpha|^2/2}-1+p_D^{(b)}\right)}\right)
\end{align*}
is a positive constant.
Since $q=\mathcal{O}(1/\sqrt{n})$, this demonstrates that 
  $\mathcal{O}(\sqrt{n})$ bits can be covertly transmitted from Alice to Bob
  with $\mathbb{P}_e^{(b)}<\delta$ for arbitrary $\delta>0$ given 
  large enough $n$.
\end{proof}

\begin{theorem}\emph{(Dark counts yield square root law under structured modulation)}
\label{th:lpd_ppm_s}
\input{th_lpd_ppm}
\end{theorem}
\begin{proof}
\textit{Construction:}
Prior to communication, Alice and Bob secretly choose a random subset 
  $\mathcal{S}$ of PPM frames to use for transmission by selecting each of 
  $n/Q$ available PPM frames independently with probability $\zeta$. 
Alice and Bob then secretly generate a vector $\mathbf{k}$ containing 
  $|\mathcal{S}|$ numbers selected independently uniformly at random from 
  $\{0,1,\ldots,Q-1\}$, where $|\mathcal{S}|$ denotes the cardinality of
  $\mathcal{S}$. 
Alice encodes a message into a codeword of size $|\mathcal{S}|$ using an ECC 
  that may be known to Willie. 
She adds $\mathbf{k}$ modulo $Q$ to this message and transmits it on the 
  PPM frames in $\mathcal{S}$. 
  
\textit{Analysis (Willie):} 
Willie detects each PPM frame received from Alice, recording the 
  photon counts in
  $\mathbf{y}_w=[\mathbf{y}_1^{(w)},\ldots,\mathbf{y}_n^{(w)}]$ where 
  $\mathbf{y}_i^{(w)}=[y_{i,1}^{(w)},\ldots,y_{i,Q}^{(w)}]$ and
  $y_{i,j}^{(w)}$ it the number of photons observed in the
  $j^{\text{th}}$ mode of the $i^{\text{th}}$ PPM frame.
Denote by $\mathbb{P}_0$ the distribution of $\mathbf{y}_w$ when Alice
  does not transmit and by $\mathbb{P}_1$ the distribution when she transmits.
When Alice does not transmit, Willie's receiver observes a Poisson dark count
  process with rate $\lambda_w$ photons per mode, implying that
  $\mathbf{y}_w$ is a vector of $nQ$ i.i.d.~$\text{Poisson}(\lambda_w)$
  random variables.
Therefore, $\{\mathbf{y}_i^{(w)}\}$ is i.i.d.~with 
  $\mathbf{y}_i^{(w)}\sim \mathbb{P}_w$ and
  $\mathbb{P}_0=\mathbb{P}_w^n$, where $\mathbb{P}_w$ is the
  distribution of $Q$ i.i.d.~$\text{Poisson}(\lambda_w)$ random variables with
  p.m.f.:
\begin{align}
\label{eq:ppm_p0}p_0(\mathbf{y}_i^{(w)})&=\prod_{j=1}^Q\frac{\lambda_w^{y_{i,j}^{(w)}}e^{-\lambda_w}}{y_{i,j}^{(w)}!}.
\end{align}

When Alice transmits, by construction, each PPM frame is randomly selected
  for transmission with probability $\zeta$.
In each selected PPM frame, a pulse is transmitted using one of $Q$ modes 
  chosen equiprobably. 
Therefore, in this case $\{\mathbf{y}_i^{(w)}\}$ is also i.i.d.~with
  $\mathbf{y}_i^{(w)}\sim \mathbb{P}_s$ and 
  $\mathbb{P}_1=\mathbb{P}_s^n$, where the p.m.f.~of $\mathbb{P}_s$ is:
\begin{align}
\label{eq:ppm_p1}p_1(\mathbf{y}_i^{(w)})&=(1-\zeta)\prod_{j=1}^Q\frac{\lambda_w^{y_{i,j}^{(w)}}e^{-\lambda_w}}{y_{i,j}^{(w)}!}+\frac{\zeta}{Q}\sum_{m=1}^Q\frac{(\eta_w|\alpha|^2+\lambda_w)^{y_{i,m}^{(w)}}e^{-\eta_w|\alpha|^2-\lambda_w}}{y_{i,m}^{(w)}!}\prod_{\myatop{j=1}{j\neq m}}^Q\frac{\lambda_w^{y_{i,j}^{(w)}}e^{-\lambda_w}}{y_{i,j}^{(w)}!}.
\end{align}

Since the classical relative entropy is additive for product distributions,
  $D(\mathbb{P}_0\|\mathbb{P}_1)=\frac{n}{Q}D(\mathbb{P}_w\|\mathbb{P}_s)$.
Now, denoting by $\mathbf{x}=[x_1,\cdots,x_Q]$ where $x_j\in\mathbb{N}_{0}$,
  we have:
\begin{align}
\label{eq:ppm_kl}D(\mathbb{P}_w\|\mathbb{P}_s)&=-\sum_{\mathbf{x}\in\mathbb{N}_0^Q}\prod_{j=1}^Q\frac{\lambda_w^{x_j}e^{-\lambda_w}}{x_j!}\log\left[1-\zeta+\frac{\zeta}{Q}\sum_{m=1}^Q\left(1+\frac{\eta_w|\alpha|^2}{\lambda_w}\right)^{x_m}e^{-\eta_w|\alpha|^2}\right]\\
&\leq\frac{\zeta^2\left(e^{(\eta_w|\alpha|^2)^2/\lambda_w}-1\right)}{2Q}\nonumber
\end{align}
where the inequality is due to the Taylor's 
  Theorem with the remainder applied to the Taylor series expansion of 
  equation \eqref{eq:ppm_kl} with respect to $\zeta$ at $\zeta=0$.
By Lemma \ref{lemma:cre_lb},
  $\zeta=\frac{4\epsilon Q}{\sqrt{n\left(e^{(\eta_w|\alpha|^2)^2/\lambda_w}-1\right)}}$
  ensures that Willie's error probability is lower-bounded by 
  $\mathbb{P}_e^{(w)}\geq\frac{1}{2}-\epsilon$.

\textit{Analysis (Bob):}
As in the proof of Theorem \ref{th:dark_s}, Bob uses a practical
  SPD receiver with probability of a dark click $p_D^{(b)}$.
Bob examines only the PPM frames in $\mathcal{S}$. 
If two or more clicks are detected in a PPM frame, a PPM symbol is assigned by
  selecting one of the clicks uniformly at random. 
If no clicks are detected, the PPM frame is labeled as an \emph{erasure}.
After subtracting $\mathbf{k}$ modulo $Q$ from this vector of PPM symbols 
  (subtraction is not performed on erasures), the resultant vector is passed 
  to the decoder. 
A random coding argument \cite[Theorem 5.6.2]{gallager68IT} yields
  reliable transmission of $\mathcal{O}\left(\sqrt{\frac{n}{Q}}\log Q\right)$ 
  covert bits. 
\end{proof}
\\

\begin{theorem}\emph{(Converse of the square root law)}\label{th:converse_s}
\input{th_converse}
\end{theorem}
\begin{proof}
As in the proof of Theorem \ref{th:pureloss_s}, Alice sends one of $2^M$ 
  (equally likely) $M$-bit messages by choosing an element from an arbitrary 
  codebook $\{\hat{\rho}^{A^n}_x,x=1,\ldots,2^M\}$, where a state 
  $\hat{\rho}^{A^n}_x=\ket{\psi_x}^{A^nA^n}\hspace{-4pt}\bra{\psi_x}$
  encodes an $M$-bit message $W_x$.
$\ket{\psi_x}^{A^n}=\sum_{\mathbf{k}\in\mathbb{N}_{0}^n}a_{\mathbf{k}}(x)\ket{\mathbf{k}}$ is a general 
  $n$-mode pure state, where
  $\ket{\mathbf{k}}\equiv\ket{k_1}\otimes\ket{k_2}\otimes\cdots\otimes\ket{k_n}$
  is a tensor product of $n$ Fock states. 
The mean photon number of a codeword $\hat{\rho}^{A^n}_x$ is
  $\bar{n}_x=\sum_{\mathbf{k}\in\mathbb{N}_{0}^n}(\sum_{i=1}^nk_i)|a_{\mathbf{k}}(x)|^2$,
  and the photon number variance is 
  $\sigma^2_x=\sum_{\mathbf{k}\in\mathbb{N}_{0}^n}(\sum_{i=1}^nk_i)^2|a_{\mathbf{k}}(x)|^2-\bar{n}_x^2=\mathcal{O}(n)$.
We limit our analysis to pure input states since, by convexity, using mixed
  states as inputs can only deteriorate the performance (since that is
  equivalent to transmitting a randomly chosen pure state from an ensemble and 
  discarding the knowledge of that choice).

Willie uses a noisy PNR receiver to observe his channel from Alice, and 
  records the total photon count $X_{\textit{tot}}$ over $n$ modes.
For some threshold $S$ that we discuss later, Willie declares that Alice 
  transmitted when $X_{\textit{tot}}\geq S$, and did not transmit when 
  $X_{\textit{tot}}<S$.
When Alice does not transmit, Willie observes noise: 
  $X_{\textit{tot}}^{(0)}=X_D+X_T$, where $X_D$ 
  is the number of dark counts due to the spontaneous emission
  process at the detector, and $X_T$ is the number of photons observed due to 
  the thermal background.
Since the dark counts are modeled by a Poisson process with rate $\lambda_w$ 
  photons per mode, both the mean and variance of the observed dark counts per
  mode is $\lambda_w$.
The mean of the number of photons observed per mode from the thermal 
  background with mean photon number per mode $\bar{n}_T$ is 
  $(1-\eta_w)\bar{n}_T$ and the variance is
  $(1-\eta_w)^2(\bar{n}_T+\bar{n}_T^2)$.
Thus, the mean of the total number of noise photons observed per 
  mode is $\mu_N=\lambda_w+(1-\eta_w) \bar{n}_T$, and, due to the statistical 
  independence of the noise processes, the variance is
  $\sigma^2_N=\lambda_w+(1-\eta_w)^2(\bar{n}_T+\bar{n}_T^2)$.
We upper-bound the false alarm probability using Chebyshev's inequality: 
\begin{align}
\mathbb{P}_{\mathrm{FA}}&=\mathbb{P}(X_{\textit{tot}}^{(0)}\geq S)\nonumber\\
\label{eq:conv_pfa}&\leq\frac{n\sigma^2_N}{(S-n\mu_N)^2},
\end{align}
where equation \eqref{eq:conv_pfa} is due to the memorylessness of the noise
  processes.
Thus, to obtain the desired $\mathbb{P}_{\mathrm{FA}}^*$, Willie sets 
  threshold $S=n\mu_N+\sqrt{n\sigma^2_N/\mathbb{P}_{\mathrm{FA}}^*}$.

When Alice transmits codeword $\hat{\rho}^{A^n}_u$ corresponding to message 
  $W_u$, Willie observes $X_{\textit{tot}}^{(1)}=X_u+X_D+X_T$, 
  where $X_u$ is the count due to Alice's transmission.
We upper-bound the missed detection probability using Chebyshev's inequality:
\begin{align}
\mathbb{P}_{\mathrm{MD}}&=\mathbb{P}(X_{\textit{tot}}^{(1)}<S)\nonumber\\
&\leq\mathbb{P}\left(|X_{\textit{tot}}^{(1)}-\eta_w\bar{n}_u-n\mu_N|\geq \eta_w\bar{n}_u-\sqrt{\frac{n\sigma^2_N}{\mathbb{P}_{\mathrm{FA}}^*}}\right)\nonumber\\
\label{eq:indep_var}&\leq \frac{n\sigma^2_N+\eta_w^2\sigma^2_u}{(\eta_w\bar{n}_u-\sqrt{n\sigma^2_N/\mathbb{P}_{\mathrm{FA}}^*})^2},
\end{align}
  where equation \eqref{eq:indep_var} is due to the independence between the 
  noise and Alice's codeword.
Since $\sigma^2_u=\mathcal{O}(n)$, if $\bar{n}_u=\omega(\sqrt{n})$, then
  $\lim_{n\rightarrow\infty}\mathbb{P}_{\mathrm{MD}}=0$.
Thus, given large enough $n$, Willie can detect Alice's codewords that have
  mean photon number $\bar{n}_u=\omega(\sqrt{n})$ with probability of
  error $\mathbb{P}_e^{(w)}\leq\epsilon$ for any $\epsilon>0$.

Now, if Alice wants to lower-bound $\mathbb{P}_e^{(w)}$, her codebook
  must contain a positive fraction of codewords with mean photon number 
  upper-bounded by $\bar{n}_{\mathcal{U}}=\mathcal{O}(\sqrt{n})$.
Formally, there must exist a subset of the codebook 
  $\left\{\hat{\rho}^{A^n}_u,u\in\mathcal{U}\right\}$, where
  $\mathcal{U}=\left\{u:\bar{n}_u\leq\bar{n}_{\mathcal{U}}\right\}$,
  with $\frac{|\mathcal{U}|}{2^M}\geq\kappa$ and $\kappa>0$.
Suppose Bob has an unattenuated pure-loss channel from Alice ($\eta_b=0$ and
  $\bar{n}_T=0$) and access to any receiver allowed by quantum mechanics.
The decoding error probability $\mathbb{P}_e^{(b)}$ in such scenario 
  clearly lower-bounds the decoding error probability in a practical scenario
  where the channel from Alice is lossy and either the channel or the receiver 
  are noisy.
Denote by $E_{a\rightarrow k}$ the event that a transmitted message $W_a$ is 
  decoded as $W_k\neq W_a$.
Since the messages are 
  equiprobable, the average probability of error for the 
  codebook containing only the codewords in $\mathcal{U}$ is: 
\begin{align}
\mathbb{P}_e^{(b)}(\mathcal{U})&=\frac{1}{|\mathcal{U}|}\sum_{a\in\mathcal{U}}\mathbb{P}\left(\cup_{k\in\mathcal{U}\backslash\{a\}}E_{a\rightarrow k}\right).
\end{align}
Since the probability that a message is sent from $\mathcal{U}$ is $\kappa$,
\begin{align}
\label{eq:conv_pe_lb}\mathbb{P}_e^{(b)}&\geq\kappa\mathbb{P}_e^{(b)}(\mathcal{U}).
\end{align}
Equality holds only when Bob receives messages that are
  not in $\mathcal{U}$ error-free and knows when the messages from 
  $\mathcal{U}$ are sent (in other words, equality holds when the set of 
  messages on which decoder can err is reduced to $\mathcal{U}$).
Denote by $W_a$, $a\in\mathcal{U}$, the message transmitted by Alice, and by 
  $\hat{W}_a$ Bob's decoding of $W_a$.
Then, since each message is equiprobable and $|\mathcal{U}|=\kappa2^M$,
\begin{align}
\log_2\kappa+M&=H(W_a)\\
\label{eq:altproof_2}&=I(W_a;\hat{W}_a)+H(W_a|\hat{W}_a)\\
\label{eq:altproof_fano}&\leq I(W_a;\hat{W}_a)+1+(\log_2 \kappa+M)\mathbb{P}_e^{(b)}(\mathcal{U})\\
\label{eq:altproof_holevo}&\leq\chi\left(\left\{\frac{1}{|\mathcal{U}|},\hat{\rho}^{A^n}_u\right\}\right)+1+(\log_2\kappa+M)\mathbb{P}_e^{(b)}(\mathcal{U})
\end{align}
where \eqref{eq:altproof_2} is from the definition of mutual information,
  \eqref{eq:altproof_fano} is due to classical Fano's inequality
  \cite[Equation~(9.37)]{cover02IT},
  and \eqref{eq:altproof_holevo} is the Holevo bound 
  $I(X;Y)\leq \chi(\{p_X(x),\hat{\rho}_x\})$~\cite{hol97}.
The mutual information $I(X;Y)$ is between a classical input $X$ and a 
  classical output $Y$, which is a function of the prior probability 
  distribution $p_X(x)$, and the conditional probability distribution 
  $p_{Y|X}(y|x)$, with $x \in {\mathcal X}$ and $y \in {\mathcal Y}$.
The classical input $x$ maps to a quantum state ${\hat \rho}_x$. 
A \emph{specific choice} of a quantum measurement, described by POVM 
  elements $\{\Pi_y, y \in {\mathcal Y}\}$, induces the conditional
  probability distribution
  $p_{Y|X}(y|x)=\trace\left[\Pi_y\hat{\rho}_x\right]$.
The Holevo information, 
  $\chi(\left\{p_x,\hat{\rho}_x\right\})=S\left(\sum_{x \in {\mathcal X}}p_x\hat{\rho}_x\right)-\sum_{x \in {\mathcal X}}p_xS(\hat{\rho}_x)$,
  where $S(\hat{\rho})\equiv-\trace[\hat{\rho}\ln\hat{\rho}]$ is the
  von Neumann entropy of the state $\hat{\rho}$, is not a function of the 
  quantum measurement.
Since 
  $\hat{\rho}^{A^n}_u=\ket{\psi_u}^{A^nA^n}\hspace{-4pt}\bra{\psi_u}$
  is a pure state, 
  $\chi\left(\{\frac{1}{|\mathcal{U}|},\hat{\rho}^{A^n}_u\}\right)=S\left(\frac{1}{|\mathcal{U}|}\sum_{u\in\mathcal{U}}\ket{\psi_u}^{A^nA^n}\hspace{-4pt}\bra{\psi_u}\right)$.
Denote the ``average codeword'' in $\mathcal{U}$ by 
  $\bar{\rho}^{A^n}=\frac{1}{|\mathcal{U}|}\sum_{u\in\mathcal{U}}\ket{\psi_u}^{A^nA^n}\hspace{-4pt}\bra{\psi_u}$,
  and the state of the $j^{\text{th}}$ mode of $\bar{\rho}^{A^n}$ by 
  $\bar{\rho}^{A^n}_j$.
We obtain $\bar{\rho}^{A^n}_j$ by tracing out all the other modes in
  $\bar{\rho}^{A^n}$ and denote its mean photon number by $\bar{n}_j$
  (i.e.~$\bar{n}_j$ is the mean photon number of the $j^{\text{th}}$ mode of 
  $\bar{\rho}^{A^n}$).
Finally, denote a coherent state ensemble with a 
  zero-mean circularly-symmetric Gaussian distribution by
  $\hat{\rho}^T_{\bar{n}}=\frac{1}{\pi \bar{n}}\int e^{-{|\alpha|^2}/{\bar{n}}}|\alpha\rangle\langle\alpha |{\rm d}^2\alpha$.
The von Neumann entropy of $\hat{\rho}^T_{\bar{n}}$, 
  $S\left(\hat{\rho}^T_{\bar{n}}\right)=\log_2(1+\bar{n})+\bar{n}\log_2\left(1+\frac{1}{\bar{n}}\right)$.
Now,
\begin{align}
\label{eq:subadditivity}S\left(\bar{\rho}^{A^n}\right)&\leq\sum_{j=1}^nS\left(\bar{\rho}^{A^n}_j\right)\\
\label{eq:gaussianmax}&\leq\sum_{i=1}^n\log_2(1+\bar{n}_j)+\bar{n}_j\log_2\left(1+\frac{1}{\bar{n}_j}\right)\\
\label{eq:concavity}&\leq n\left(\log_2\left(1+\frac{\bar{n}_{\mathcal{U}}}{n}\right)+\frac{\bar{n}_{\mathcal{U}}}{n}\log_2\left(1+\frac{n}{\bar{n}_{\mathcal{U}}}\right)\right),
\end{align}
where \eqref{eq:subadditivity} follows from the sub-additivity of the 
  von Neumann entropy and \eqref{eq:gaussianmax} is due to 
  $\hat{\rho}^T_{\bar{n}}$
  maximizing the von Neumann entropy of a single-mode state with mean photon
  number constraint $\bar{n}$~\cite{giovannetti04cappureloss}.
Now, $S\left(\hat{\rho}^T_{\bar{n}}\right)$ is concave and increasing for 
  $\bar{n}>0$,
  and, since $\sum_{j=1}^n\bar{n}_j\leq\bar{n}_{\mathcal{U}}$ by construction of
  $\mathcal{U}$, the application of Jensen's inequality yields
  \eqref{eq:concavity}.
Combining \eqref{eq:altproof_holevo} and \eqref{eq:concavity} 
  and solving for $\mathbb{P}_e^{(b)}(\mathcal{U})$ yields:
\begin{align}
\label{eq:conv_peu_b}\mathbb{P}_e^{(b)}(\mathcal{U})&\geq1-\frac{\log_2\left(1+\frac{\bar{n}_{\mathcal{U}}}{n}\right)+\frac{\bar{n}_{\mathcal{U}}}{n}\log_2\left(1+\frac{n}{\bar{n}_{\mathcal{U}}}\right)+\frac{1}{n}}{\frac{\log_2\kappa}{n}+\frac{M}{n}}.
\end{align}
Substituting \eqref{eq:conv_peu_b} into \eqref{eq:conv_pe_lb} 
  yields the following lower bound on Bob's decoding error probability:
\begin{align}
\label{eq:conv_pe_b}\mathbb{P}_e^{(b)}&\geq\kappa\left[1-\frac{\log_2\left(1+\frac{\bar{n}_{\mathcal{U}}}{n}\right)+\frac{\bar{n}_{\mathcal{U}}}{n}\log_2\left(1+\frac{n}{\bar{n}_{\mathcal{U}}}\right)+\frac{1}{n}}{\frac{\log_2\kappa}{n}+\frac{M}{n}}\right]
\end{align}
Since Alice transmits $\omega(\sqrt{n})$ bits in $n$ modes,
  $M/n=\omega(1/\sqrt{n})$ bits/symbol.
However, since $\bar{n}_{\mathcal{U}}=\mathcal{O}(\sqrt{n})$, as
  $n\rightarrow\infty$, $\mathbb{P}_e^{(b)}$ is bounded away from 
  zero for any $\kappa>0$.
Thus, Alice cannot transmit $\omega(\sqrt{n})$ bits in $n$ optical modes 
  both covertly and reliably.
\end{proof}

\section{Details of the experimental methods}

Here we provide the mathematical details of the experimental methods.

\subsection*{Calculation of Bob's maximum throughput}
The $Q$-ary PPM signaling combined with Bob's 
  device for assigning symbols to received PPM frames induces a discrete
  memoryless channel described by a conditional distribution $\mathbb{P}(Y|X)$,
  where $X\in\{1,\ldots,Q\}$ is Alice's input symbol and 
  $Y\in\{1,\ldots,Q,\mathcal{E}\}$ is Bob's output symbol with $\mathcal{E}$ 
  indicating an erasure.
Since Bob observes Alice's pulse with probability 
  $1-e^{-\bar{n}^{(b)}_{det}}$,
  $\mathbb{P}(Y|X)$ is characterized as follows:
\begin{align*}
&\mathbb{P}(Y=x|X=x)=\left(1-e^{-\bar{n}^{(b)}_{det}}\right)\sum_{i=0}^{Q-1}\frac{1}{i+1}\left(p_D^{(b)}\right)^i\left(1-p_D^{(b)}\right)^{Q-1-i}+e^{-\bar{n}^{(b)}_{det}}\sum_{i=1}^Q\frac{1}{i}\left(p_D^{(b)}\right)^i\left(1-p_D^{(b)}\right)^{Q-i}\\
&\mathbb{P}(Y=\mathcal{E}|X=x)=e^{-\bar{n}^{(b)}_{det}}\left(1-p_D^{(b)}\right)^{Q}\\
&\mathbb{P}(Y=y,y\notin \{x,\mathcal{E}\}|X=x)=\frac{1-\mathbb{P}(Y=x|x)-\mathbb{P}(Y=\mathcal{E}|x)}{Q-1}
\end{align*}
The symmetry of this channel yields the Shannon capacity~\cite{shannon48it}
  $C_s=I(X;Y)$, where $\mathbb{P}(X=x)=\frac{1}{Q}$ for $x=1,\ldots,Q$
  and $I(X;Y)$ is the mutual information between $X$ and $Y$.
We use the experimentally-observed values from Table \ref{tab:exp_params} 
  to compute $C_s$ 
  for each regime, and plot $\frac{C_s\zeta n}{Q}$ since $\frac{\zeta n}{Q}$ 
  is the expected number of PPM frames selected for transmission.

\subsection*{Derivation of the log-likelihood ratio test statistic}
The log-likelihood ratio test statistic is given by 
  $L=\frac{f_1(\mathbf{x}_w)}{f_0(\mathbf{x}_w)}$, where $f_0(\mathbf{x}_w)$ and
  $f_1(\mathbf{x}_w)$ are the likelihood functions of the click record 
  $\mathbf{x}_w$ corresponding to Alice being quiet and transmitting.
Click record $\mathbf{x}_w$ contains Willie's observations of each PPM frame
  on his channel from Alice
  $\mathbf{x}_w=[\mathbf{x}_1^{(w)},\ldots,\mathbf{x}_{n/Q}^{(w)}]$. 
$\mathbf{x}_i^{(w)}=[x_{i,1}^{(w)},\ldots,x_{i,Q}^{(w)}]$ contains the
  observation of the $i^{\text{th}}$ PPM frame with 
  $x_{i,j}^{(w)}\in\{0,1\}$, where ``0'' and ``1'' indicate the absence and the
  presence of a click, respectively.
When Alice does not transmit, Willie only observes dark clicks.
Thus, each $\mathbf{x}_w$ is a vector of
  i.i.d.~$\text{Bernoulli}\left(p_D^{(w)}\right)$ random variables.
The likelihood function of $\mathbf{x}_w$ under $H_0$ is then: 
\begin{align*}
f_0(\mathbf{x}_w)&=\prod_{i=1}^{n/Q}\left(p_D^{(w)}\right)^{\sum_{j=1}^Qx_{i,j}^{(w)}}\left(1-p_D^{(w)}\right)^{Q-\sum_{j=1}^Qx_{i,j}^{(w)}}.
\end{align*}
  
Now consider the scenario when Alice transmits.
The secret shared between Alice and Bob identifies the random subset 
  $\mathcal{S}$ of the PPM frames used for transmission, and a random vector 
  $\mathbf{k}$ which is modulo-added to the codeword. 
Modulo addition of $\mathbf{k}$ effectively selects a random pulse location
  within each PPM frame.
Note that, while both the construction in the proof of Theorem \ref{th:lpd_ppm}
  and Alice's encoder described in the Methods generate $\mathcal{S}$ first
  and then $\mathbf{k}$, the order of these operations can be reversed: 
  we can first fix a random location of a pulse in each of $n/Q$ PPM
  frames, and then select a random subset of these frames.
Consider Willie's observation of the $i^{\text{th}}$ PPM frame, and suppose 
  that the $l^{\text{th}}$ mode is used if the frame is selected for 
  transmission.
Denote the probability of Willie's detector observing Alice's pulse by
  $p_r^{(w)}=1-e^{-\bar{n}^{(w)}_{\mathit{det}}}$.
By construction, frames are selected for 
  transmission independent of each other with probability $\zeta$.
Willie's detector registers a click on this mode when one of the following 
  disjoint events occurs:
\begin{itemize}
\item{PPM frame is selected and pulse is detected by Willie 
  (probability $\zeta p_r^{(w)}$);}
\item{PPM frame is selected, but Willie, instead of detecting it, records 
  a dark click (probability $\zeta\left(1-p_r^{(w)}\right)p_D^{(w)}$); and,}
\item{PPM frame is not selected, but Willie records a dark click 
  (probability $(1-\zeta)p_D^{(w)}$).}
\end{itemize}
The probability of the union of these events is
\begin{align}
\label{eq:p_s}p_s^{(w)}&=\zeta p_r^{(w)}\left(1-p_D^{(w)}\right)+p_D^{(w)},
\end{align}
and, thus, Willie observes an independent 
  $\text{Bernoulli}\left(p_s^{(w)}\right)$ random 
  variable in the $l^{\text{th}}$ mode of the $i^{\text{th}}$ PPM frame.
Since Alice only uses the $l^{\text{th}}$ mode for transmission, in modes other 
  than the $l^{\text{th}}$, Willie observes a set of 
  $Q-1$  i.i.d.~$\text{Bernoulli}\left(p_D^{(w)}\right)$ random variables 
  corresponding to dark clicks.
The click record $\mathbf{x}_i^{(w)}$ of the $i^{\text{th}}$ PPM frame is 
  independent of other PPM frame and, thus, the 
  likelihood function of $\mathbf{x}_w$ under $H_1$ is
  $\mathbb{P}_1(\mathbf{x}_w=\mathbf{x})=\prod_{i=1}^np_1(\mathbf{x}_i^{(w)})$
  where
\begin{align*}
f_1(\mathbf{x}_w)&=\prod_{i=1}^{n/Q}\frac{1}{Q}\sum_{l=1}^Q\left(p_s^{(w)}\right)^{x_{i,l}^{(w)}}\left(1-p_s^{(w)}\right)^{1-x_{i,l}^{(w)}}\left(p_D^{(w)}\right)^{\sum_{\myatop{j=1}{j\neq l}}^Qx_{i,j}^{(w)}}\left(1-p_D^{(w)}\right)^{Q-1-\sum_{\myatop{j=1}{j\neq l}}^Qx_{i,j}^{(w)}}.
\end{align*}
Evaluation of the likelihood ratio yields:
\begin{align}
\nonumber\frac{f_1(\mathbf{x}_w)}{f_0(\mathbf{x}_w)}&=\prod_{i=1}^{n/Q}\frac{1-\zeta p_r^{(w)}}{Q}\sum_{l=1}^Q\left(1+\frac{\zeta p_r^{(w)}}{\left(1-\zeta p_r^{(w)}\right)p_D^{(w)}}\right)^{x_{i,l}^{(w)}}\\
\label{eq:simplify_ratio}&=\prod_{i=1}^{n/Q}\left[1-\zeta p_r^{(w)}+\frac{\zeta p_r^{(w)}y_i^{(w)}}{Qp_D^{(w)}}\right]
\end{align}
where $y_i^{(w)}=\sum_{l=1}^Qx_{i,l}^{(w)}$ and the simplification 
  yielding \eqref{eq:simplify_ratio} is due to $x_{i,l}^{(w)}\in\{0,1\}$.
Taking the logarithm of equation \eqref{eq:simplify_ratio} yields
  the log-likelihood ratio in \eqref{eq:pure_LRT}.

\subsection*{Gaussian approximation of $\mathbb{P}_e^{(w)}$}
For small $\zeta$, the Taylor series expansion
  of the summand in \eqref{eq:pure_LRT} at $\zeta=0$, 
  $\log\left[1+\zeta p_r^{(w)}\left(\frac{y_i^{(w)}}{Qp_D^{(w)}}-1\right)\right]
  \approx\zeta p_r^{(w)}\left(\frac{y_i^{(w)}}{Qp_D^{(w)}}-1\right)$,
  yields an approximation for the log-likelihood ratio:
$L\approx\frac{\zeta p_r^{(w)}}{Qp_D^{(w)}}\left(\sum_{i=1}^{n/Q}y_i^{(w)}-np_D^{(w)}\right)$.
Thus, effectively, Willie uses the total click count 
  $Y=\sum_{i=1}^{n/Q}y_i^{(w)}$ as a
  test statistic, which explains the lack of sensitivity of our test to 
  the variations in the observed channel characteristics given in Table 
  \ref{tab:exp_params}.

This also provides an analytical approximation of $\mathbb{P}_e^{(w)}$.
First consider the case when Alice is not transmitting.
Then the total click count is a binomial random variable 
  $Y\sim\mathcal{B}\left(y;p_D^{(w)},n\right)$ whose distribution, for large 
  $n$, can be approximated using the central limit theorem
  by a Gaussian distribution $\Phi\left(y;\mu_0,\sigma^2_0\right)$
  with $\mu_0=np_D^{(w)}$ and $\sigma^2_0=np_D^{(w)}\left(1-p_D^{(w)}\right)$,
  where $\Phi\left(x;\mu,\sigma^2\right)=\frac{1}{\sqrt{2\pi}\sigma}\int_{-\infty}^xe^{-\frac{|t-\mu|^2}{2\sigma^2}}dt$
  is the distribution function of a Gaussian random variable 
  $\mathcal{N}(x;\mu,\sigma^2)$.
Now consider the case when Alice is transmitting.
Since $\mathcal{S}$ and $\mathbf{k}$ are unknown to Willie, the total click
  count is the sum of two independent but not identical binomial random
  variables 
  $Y=X+Z$, where $X\sim\mathcal{B}\left(x;p_D^{(w)},n-\frac{n}{Q}\right)$ is the
  number of dark clicks in the $n-\frac{n}{Q}$ modes that Alice never uses in 
  a PPM scheme and $Z\sim\mathcal{B}\left(z;p_s^{(w)},\frac{n}{Q}\right)$ is the
  contribution from the $\frac{n}{Q}$ modes that Alice can use to transmit,
  with $p_s^{(w)}$ defined in \eqref{eq:p_s}.
By the central limit theorem, for large $n$, the distribution of $X$ can be 
  approximated using a Gaussian distribution $\Phi(x;\mu_X,\sigma^2_X)$ where
  $\mu_X=\left(n-\frac{n}{Q}\right)p_D^{(w)}$ and 
  $\sigma^2_X=\left(n-\frac{n}{Q}\right)p_D^{(w)}\left(1-p_D^{(w)}\right)$.
Similarly, the distribution of $Z$ can be approximated by a Gaussian 
  distribution $\Phi\left(z;\mu_Z,\sigma^2_Z\right)$ where
  $\mu_Z=\frac{n}{Q}\left(\zeta p_r^{(w)}+\left(1-\zeta p_r^{(w)}\right)p_D^{(w)}\right)$ and 
  $\sigma^2_Z=\frac{n}{Q}\left(\zeta p_r^{(w)}+\left(1-\zeta p_r^{(w)}\right)p_D^{(w)}\right)\left(1-\zeta  p_r^{(w)}\right)\left(1-p_D^{(w)}\right)$.
Thus,  the distribution of $Y$ can be approximated by a Gaussian distribution 
  $\Phi\left(y;\mu_1,\sigma^2_1\right)$
  with $\mu_1=\mu_X+\mu_Z$ and $\sigma^2_1=\sigma^2_X+\sigma^2_Z$
  via the additivity property of independent Gaussian random variables.
Willie's probability of error is thus approximated by:
\begin{align}
\label{eq:pe_approx}\tilde{\mathbb{P}}_e^{(w)}=\frac{1}{2}\min_S(1-\Phi(S;\mu_0,\sigma^2_0)+\Phi(S;\mu_1,\sigma^2_1)).
\end{align}
The value of the threshold $S^*$ that minimizes the RHS of \eqref{eq:pe_approx} 
  satisfies
  $\frac{|S^*-\mu_0|^2}{\sigma^2_0}-\log(\sigma^2_1/\sigma^2_0)=\frac{|S^*-\mu_1|^2}{\sigma^2_1}$.